\documentclass[english,cleveref,thm-restate]{lipics-v2021}

\hideLIPIcs  

\title{Structure and Independence in Hyperbolic Uniform Disk Graphs}

\author{Thomas Bläsius}{Karlsruhe Institute of Technology, Karlsruhe,
  Germany \and
  \url{http://scale.iti.kit.edu}}{thomas.blaesius@kit.edu}{https://orcid.org/0000-0003-2450-744X}{%
}

\author{Jean-Pierre {von der Heydt}}{Karlsruhe Institute of Technology, Karlsruhe, Germany \and \url{http://scale.iti.kit.edu}}{heydt@kit.edu}{https://orcid.org/0009-0000-3852-350X}{Funded by the pilot program Core-Informatics of the Helmholtz Association (HGF).}

\author{Sándor Kisfaludi-Bak}{Aalto University, Espoo, Finland \and \url{}}{sandor.kisfaludi-bak@aalto.fi}{https://orcid.org/0000-0002-6856-2902}{Supported by the Research Council of Finland, Grant 363444.}

\author{Marcus Wilhelm}{Karlsruhe Institute of Technology, Karlsruhe, Germany \and \url{http://scale.iti.kit.edu}}{marcus.wilhelm@kit.edu}{https://orcid.org/0000-0002-4507-0622}{funded by the Deutsche Forschungsgemeinschaft (DFG, German Research Foundation) -- 524989715.}

\author{Geert {van Wordragen}}{Aalto University, Espoo, Finland \and
  \url{}}{geert.vanwordragen@aalto.fi}{https://orcid.org/0000-0002-2650-638X}{}

\authorrunning{T. Bläsius, S. Kisfaludi-Bak, G. van Wordragen, J.-P. von der Heydt, M. Wilhelm} 

\Copyright{Thomas Bläsius, Jean-Pierre von der Heydt, Sándor
  Kisfaludi-Bak, Marcus Wilhelm,\\ Geert van Wordragen}

\ccsdesc[300]{Theory of computation~Design and analysis of algorithms}
\ccsdesc[500]{Theory of computation~Computational geometry}
\ccsdesc[500]{Theory of computation~Graph algorithms analysis}

\keywords{hyperbolic geometry, unit disk graphs, independent set, treewidth} 

\category{} 

\relatedversiondetails[cite=Struc_Indep_Hyper_Unifo_pre2024]{Full Version}{https://arxiv.org/abs/2407.09362}

\nolinenumbers 

\usepackage{uniinput}  
\newcommand{\eps}{\varepsilon}
\newcommand{\sig}{\varsigma}

\newcommand{\bags}{\mathrm{Bags}}
\newcommand{\Hyp}{\mathbb{H}}

\newcommand{\twformulak}{\ensuremath{\frac{\log k}{r}}}
\newcommand{\Oh}{\mathcal{O}}
\newcommand{\norm}[1]{\left\lVert#1\right\rVert}

\newcommand{\arcosh}{\mathrm{arcosh}}

\DeclareMathOperator{\dist}{\mathrm{dist}}
\DeclareMathOperator{\area}{\mathrm{Area}}

\newcommand{\HUDG}{\mathrm{HUDG}}
\newcommand{\UDG}{\mathrm{UDG}}
\newcommand{\DG}{\mathrm{DG}}

\newcounter{ctr}
\loop
 \stepcounter{ctr}
 \expandafter\edef\csname c\Alph{ctr}\endcsname{\noexpand\mathcal{\Alph{ctr}}}
\ifnum\thectr<26
\repeat

\usepackage{tikz}
\usepackage[framemethod=tikz]{mdframed}
\usepackage{xspace}
\usepackage[disable,size=tiny]{todonotes}

\EventEditors{John Q. Open and Joan R. Access}
\EventNoEds{2}
\EventLongTitle{42nd Conference on Very Important Topics (CVIT 2016)}
\EventShortTitle{CVIT 2016}
\EventAcronym{CVIT}
\EventYear{2016}
\EventDate{December 24--27, 2016}
\EventLocation{Little Whinging, United Kingdom}
\EventLogo{}
\SeriesVolume{42}
\ArticleNo{23}
\EventEditors{Oswin Aichholzer and Haitao Wang}
\EventNoEds{2}
\EventLongTitle{41st International Symposium on Computational Geometry (SoCG 2025)}
\EventShortTitle{SoCG 2025}
\EventAcronym{SoCG}
\EventYear{2025}
\EventDate{June 23--27, 2025}
\EventLocation{Kanazawa, Japan}
\EventLogo{socg-logo.pdf}
\SeriesVolume{332}
\ArticleNo{XX}     

\begin{document}

\maketitle

\begin{abstract}
  We consider intersection graphs of disks of radius $r$ in the
  hyperbolic plane. Unlike the Euclidean setting, these graph classes
  are different for different values of $r$, where very small $r$
  corresponds to an almost-Euclidean setting and
  $r \in \Omega(\log n)$ corresponds to a firmly hyperbolic
  setting. We observe that larger values of $r$ create simpler graph
  classes, at least in terms of separators and the computational
  complexity of the \textsc{Independent Set} problem.

  First, we show that intersection graphs of disks of radius $r$ in
  the hyperbolic plane can be separated with $\Oh((1+1/r)\log n)$
  cliques in a balanced manner.  Our second structural insight
  concerns Delaunay complexes in the hyperbolic plane and may be of
  independent interest. We show that for any set $S$ of $n$ points
  with pairwise distance at least $2r$ in the hyperbolic plane, the
  corresponding Delaunay complex has outerplanarity
  $1+\Oh(\frac{\log n}{r})$, which implies a similar bound on the
  balanced separators and treewidth of such Delaunay complexes.

  Using this outerplanarity (and treewidth) bound we prove that
  \textsc{Independent Set} can be solved in
  $n^{\Oh(1+\frac{\log n}{r})}$ time. The algorithm is based on
  dynamic programming on some unknown sphere cut decomposition that is
  based on the solution. The resulting algorithm is a far-reaching
  generalization of a result of Kisfaludi-Bak (SODA 2020), and it is
  tight under the Exponential Time Hypothesis. In particular,
  \textsc{Independent Set} is polynomial-time solvable in the firmly
  hyperbolic setting of $r\in \Omega(\log n)$. Finally, in the case
  when the disks have ply (depth) at most $\ell$, we give a PTAS for
  \textsc{Maximum Independent Set} that has only quasi-polynomial
  dependence on $1/\eps$ and $\ell$. Our PTAS is a further
  generalization of our exact algorithm.
\end{abstract}

\section{Introduction}
\label{sec:introduction}
Given a set of disks in the plane, one can assign to them a
\emph{geometric intersection graph} whose vertices are the disks, and
edges are added between pairs of intersecting disks. The study of
intersection graphs is usually motivated by physically realized
networks: such networks require spatial proximity between vertices for
successful connections, and the simplest model allows for connections
within a distance $2$, which is equivalent to the intersection graph
induced by disks of radius $1$ centered at the vertices. Such graphs are called \emph{unit disk graphs}.

Unit disk graphs have received a lot of attention in the theoretical computer science literature. Their intriguing structural properties yield profound mathematical insights and facilitate the development of efficient algorithms.  Historically, unit disk graphs have been studied in the Euclidean plane, where they are well motivated due to their relevance for, e.g., sensor networks. The relevance of unit disk graphs in the hyperbolic plane is less obvious.  However, originating in the network science community, it has been observed~\cite{Hyper_Geomet_Compl_Networ-Kriouk10} that the intersection graph of randomly sampled disks of equal radius $r$ yields graphs that resemble complex real-world networks in regards to important properties\footnote{For this to work, the choice of the radius $r$ is crucial.  It is chosen as $r = \log n + c$ for a constant $c$ that controls the average degree, and the disk centers are all sampled within a disk of radius $2r$.}. They are, e.g., heterogeneous with a degree distribution following a power law \cite{Random_Hyper_Graph-Gugel12}, have high clustering coefficient~\cite{Random_Hyper_Graph-Gugel12}, and exhibit the small-world property~\cite{Diamet_Hyper_Random_Graph-FriedKrohm18,Diamet_KPKVB_Random_Graph-MüllerStaps19}.

Besides numerous structural results, these hyperbolic random graphs also allow for the design of more efficient algorithms~\cite{BlasiusFK16, Effic_Short_Paths_Scale_jour2022, BlasiusFFK23, BlasiusFK23}.  However, all these results rely on the
fact that the disks are chosen randomly.  So far, there is only little
research on hyperbolic uniform disk graphs from a deterministic, more
graph-theoretic perspective. It is important to note that for the intersection graphs of radius-$r$ disks the resulting graph classes are different for different values of $r$. Moreover, it is natural to allow for the radius $r$ to be a (monotone) function of the number of vertices. Unlike the Euclidean setting where a simple scaling shows that the choice of ``unit'' does not matter, it is a crucial parameter in the hyperbolic setting.

For those unfamiliar with hyperbolic geometry there is a way to conceptually understand these graphs in a Euclidean setting as follows. One can think of a hyperbolic disk graph of radius $r$ disks as a Euclidean disk graph where the radius of a disk of center $p$ is set to $f_r(\dist(o,p))$, where $f_r$ is decreasing very quickly at a rate set by $r$, and $\dist(o,p)$ is the distance between the origin and the disk center $p$. In particular, as $r$ goes to $0$, the rate of decrease is negligible, and we get almost Euclidean unit disk graphs, while large $r$ corresponds to a very different graph class. See Figure~\ref{fig:impact-of-radius} for an illustration.

This observation is formalized by stating that hyperbolic disk graphs of uniform (that is, equal) radius $1/n^3$ and uniform radius $\log n$ are incomparable: The former class contains the $k\times k$ grid graph with $n = k^2$ for all $k$, but does not contain stars of size $n \geq 8$.  The latter contains all star graphs but does not contain any $k\times k$ grid of size $k\geq 5$. See \Cref{fig:impact-of-radius} for an illustration and \Cref{thm:gridandstar} for a formal proof of these claims. We call the regime $r\leq 1/\sqrt{n}$ \emph{almost Euclidean}, while $r\in\Omega(\log n)$ is called \emph{firmly hyperbolic}.

\begin{figure}[tbp]
  \centering
  \hfill
  \begin{subfigure}[t]{0.31\textwidth}
    \centering
    \includegraphics[page=1]{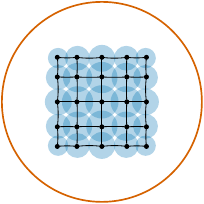}
    \subcaption{Smaller $r$ enables larger grids.}
    \label{fig:impact-of-radius-grid}
  \end{subfigure}
  \hfill
  \begin{subfigure}[t]{0.31\textwidth}
    \centering
    \includegraphics[page=2]{figures/impact-of-radius.pdf}
    \captionsetup{textformat=simple}
    \subcaption{Only small stars for small $r$.}
    \label{fig:impact-of-radius-small-star}
  \end{subfigure}
  \hfill
  \begin{subfigure}[t]{0.31\textwidth}
    \centering
    \includegraphics[page=3]{figures/impact-of-radius.pdf}
    \captionsetup{textformat=simple}
    \subcaption{Larger $r$ enables larger stars.}
    \label{fig:impact-of-radius-big-star}
  \end{subfigure}
  \hfill
  \caption{Realizing a grid is only possible for small radii, while large stars are only possible in the firmly hyperbolic setting of $r \in \Omega(\log n)$.}%
  \label{fig:impact-of-radius}
\end{figure}

The primary goal of our paper is to explore the following:

\mdfdefinestyle{highlight}{backgroundcolor=gray!20,roundcorner=10pt}
\begin{mdframed}[style=highlight]
How does the structure of hyperbolic uniform disk graphs depend on the radius $r=r(n)$?
\end{mdframed}

We note that one of the main challenges for large radii is that unlike the Euclidean plane, the kissing number\footnote{The kissing number of radius $r$ disks is the maximum number of interior-disjoint disks of radius $r$ that can touch a fixed disk of radius $r$.} of hyperbolic disks of large radius is unbounded, and as a result, there may be arbitrarily many ``nearby disks'' to any given disks in the independent set. At the same time the larger the hyperbolic disks are, the harder it is to ``surround'' them with equal radius disks, which intuitively leads to a simpler structure.

The case of constant radius hyperbolic disks (with constant kissing number) has been studied by Kisfaludi-Bak~\cite{Kisfaludi-Bak20}, but these methods fail for super-constant kissing number. The case of radius roughly $\log n$ has only been studied in a more specialized setting (inspired by hyperbolic random graph models) in~\cite{Blasius0KS23}. They call the relevant graph class \emph{strongly hyperbolic uniform disk graph}s. The above-mentioned hyperbolic random graphs are randomly sampled strongly hyperbolic uniform disk graphs. Apart from these particular choices of the radius, we have very limited understanding about the structure of hyperbolic uniform disk graphs, and the results of \cite{Blasius0KS23} do not imply anything for our hyperbolic uniform disk graphs.

The radius' impact on problem complexity was studied in the context of the traveling salesman problem, where one can use the minimum distance between points as a parameter with a similar flavor. In the hyperbolic plane, Kisfaludi-Bak~\cite{Kisfaludi-Bak21} showed that the complexity of TSP decreases as one increases the minimum pairwise distance $r$ among input points. In this paper, we show that the structure of hyperbolic uniform disk graphs behaves similarly: their structure becomes easier for larger values of $r$, and \textsc{Independent Set} can be solved faster for larger $r$.

Let $\HUDG(r(n))$ denote the set of intersection graphs where each
$n$-vertex graph can be realized as the intersection graph of disks of
uniform (equal) radius $r(n)$ in the hyperbolic plane of Gaussian
curvature\footnote{Equivalently, one can fix the radius to be $1$ and
  set the Gaussian curvature to be $-r^2$.} $-1$. We denote by $\HUDG$
the union of these classes for all radii, that is,
$\HUDG=\bigcup_{\rho>0}\HUDG(\rho)$. Let $\UDG$ denote the class of
unit disk graphs in the Euclidean plane. The authors
in~\cite{Blasius0KS23} have shown that $\UDG\subset \HUDG$. When a
graph is given without the geometric realization, it is NP-hard (even
$\exists\mathbb{R}$-complete) to decide if the graph is a
$\UDG$~\cite{McDiarmidM13} or $\HUDG$~\cite{BiekerBDJ03}; for this
reason, we will assume throughout this article that the input
intersection graphs of our algorithms are given by specifying their
geometric realization. In the hyperbolic setting, one can specify the
disk centers using a so-called \emph{model} of the hyperbolic plane,
which is simply an embedding of the hyperbolic plane into some
specific part of the Euclidean plane, e.g., inside the Euclidean unit
disk. See our preliminaries for a brief introduction to Euclidean models
of the hyperbolic plane.

From the perspective of graph algorithms, unit disk graphs in the Euclidean plane have been serving an important role as a graph class that is not comparable, but similar to planar graphs. The circle packing theorem~\cite{koebe1936kontaktprobleme} states that every planar graph can be realized as an intersection graph of disks (of arbitrary radii); thus, disk graphs serve as a common generalization of planar and unit disk graphs. Using \emph{conformal} hyperbolic models, one can observe that hyperbolic disks appear as Euclidean disks\footnote{In the Poincar\'e disk and half-plane models, all hyperbolic disks are Euclidean disks, but the Euclidean and hyperbolic radii of these disks are different.} in the model, which means that all $\HUDG$s are realized as Euclidean disk graphs. Denoting Euclidean disk graphs with $\DG$, we thus have $\UDG\subseteq \HUDG \subseteq \DG$\footnote{With a little effort, one can show that both containments are strict: $\UDG\subsetneq \HUDG \subsetneq \DG$.}.

A key structural tool in the study of both planar and (unit) disk graphs has been \emph{separator theorems}. By a result of Lipton and Tarjan\cite{LiptonT79}, any $n$-vertex planar graph can be partitioned into three vertex sets, $A,B$, and the separator $S$, such that no edges go between $A$ and $B$, the separator set $S$ has size $\Oh(\sqrt{n})$, and $\max\{|A|,|B|\}\leq \frac{2}{3}n$.
Separator theorems are closely related to the \emph{treewidth}\footnote{See \Cref{sec:prelim} for the definitions of treewidth and $\cP$-flattened treewidth.\label{ftn:treewidthdef}} of graphs~\cite{Bodlaender98}, for example, the above separator implies a treewidth bound of $\Oh(\sqrt{n})$ on planar graphs.

For (unit) disk graphs, similar separators and treewidth bounds are not possible, because one can represent cliques of arbitrary size. There have been at least three different approaches to deal with large cliques. The first option is to bound the cliques in some way. One can, for example, obtain separators for a set of disks of bounded \emph{ply}, i.e., assuming that each point of the ambient space is included in at most $\ell$ disks. There are several separators for disks (and even balls in higher dimensions) that involve ply. The strongest and most general among these is by Miller, Teng, Thurston, and Vavasis~\cite{MillerTTV97}. The second way to deal with large cliques is to use so-called clique-based separators~\cite{BergBKMZ20,BergKMT23}, where the separator is decomposed into cliques and the cliques are sometimes assigned some small weight depending on their size. In the hyperbolic setting, Kisfaludi-Bak~\cite{Kisfaludi-Bak20} showed that for any constant $c$, the graph class $\HUDG(c)$ admits a balanced separator consisting of $\Oh(\log n)$ cliques. The same paper shows that $\HUDG(c)$ also has clique-based separators of weight $\Oh(\log^2 n)$ and extends these techniques to balls of constant radius in higher-dimensional hyperbolic spaces, along with much of the machinery of~\cite{BergBKMZ20}. The third option is to use random disk positions to disperse large cliques with high probability. Bl\"asius, Friedrich, and Krohmer~\cite{BlasiusFK16} gave high-probability bounds on separators and treewidth for certain radius regimes in hyperbolic random graphs.

The above separators and treewidth can be used to obtain fast divide-and-conquer or dynamic programming algorithms~\cite{LiptonT80, BergBKMZ20}. In the \textsc{Independent Set} problem, the goal is to decide if there are $k$ pairwise non-adjacent vertices in the graph. In general graphs, the best known algorithms run in $2^{\Oh(n)}$ or $n^{\Oh(k)}$ time, and these running times are optimal under the Exponential Time Hypothesis (ETH)~\cite{fptbook} (see \cite{ImpagliazzoP01} for the definition of ETH). In planar and disk graphs however, the separators allow algorithms with running time $2^{\Oh(\sqrt{n})}$ or $n^{\Oh(\sqrt{k})}$~\cite{fptbook,MarxP22}.
Moreover, the separator implies that planar graphs have treewidth\textsuperscript{\ref{ftn:treewidthdef}} $\Oh(\sqrt{n})$ and unit disk graphs have so-called \emph{$\cP$-flattened treewidth} $\Oh(\sqrt{n})$ for some clique-partition $\cP$. As a consequence, one can adapt treewidth-based algorithms~\cite{fptbook} to these graph classes. In the hyperbolic setting, the ($\cP$-flattened) treewidth bounds yield an $n^{\Oh(\log n)}$ quasi-polynomial algorithm for \textsc{Independent Set} and many other problems on graphs in $\HUDG(c)$ for any constant $c$, which is matched by a lower bound under ETH~\cite{Kisfaludi-Bak20} in case of \textsc{Independent~Set}.

When approximating \textsc{Independent Set}, the main algorithmic tool in planar graphs is Baker's shifting technique~\cite{Baker94} that gives a $(1-\eps)$-approximate independent set in $2^{\Oh(1/\eps)}n$ time. A conceptually easier grid shifting was discovered for unit disks by Hochbaum and Maas~\cite{HochbaumM85}, which was later (implicitly) improved by Agarwal, van Kreveld, and Suri~\cite{AgarwalKS98}. This line of research culminated in the algorithm of Chan~\cite{Chan03} with a running time of $n^{\Oh(1/\eps)}$ for disks, which also can be generalized for balls in higher dimensions. Baker's algorithm and Chan's algorithm (even for unit disks) are optimal under standard complexity assumptions~\cite{Marx07a}. To the best of our knowledge, there are no published approximation algorithms for $\HUDG(r(n))$ or $\HUDG$ other than what is already implied from the algorithms on disk graphs. However, for hyperbolic random graphs,
there is an $\Oh(m \log n)$ algorithm that yields a $1 - o(1)$-approximation for \textsc{Independent Set} asymptotically almost surely~\cite{BlasiusFK23}.

An interesting setting, where the treewidth of planar graphs has been used for intersection graphs, is that of covering and packing problems in the work of Har-Peled and later Marx and Pilipczuk~\cite{Har-Peled14,MarxP22}. For the \textsc{Independent Set} problem on unit disk graphs, they consider the \emph{Voronoi diagram} of the disk centers of a maximum independent set of size $k$. For a set of points (often called \emph{sites}) $S$, the Voronoi diagram is a planar subdivision (a planar graph) that partitions the plane according to the closest site of $S$, that is, the \emph{cell} of a site $s\in S$ consists of those points whose nearest neighbor from $S$ is $s$. As the Voronoi diagram is a planar graph on $\Oh(k)$ vertices, it has a balanced separator of size $\Oh(\sqrt{k})$. Although this Voronoi diagram is based on the solution, and thus unknown from the algorithmic perspective, the idea of Marx and Pilipczuk is that one can enumerate all possible separators of the Voronoi diagram in $n^{\Oh(\sqrt{k})}$ time, and one of these separators can then be used to separate the (still unknown) solution in a balanced manner. A similar ``separator guessing'' technique has been used elsewhere in approximation and parameterized algorithms~\cite{MarxPP18,BringmannKPL19}.

\subparagraph*{Our contribution.}

Our first structural result is a balanced clique-based separator theorem presented in Section~\ref{sec:sep}. We note that all of our results in hyperbolic uniform disk graphs assume that the graph is given via some geometric representation, i.e., with the Euclidean coordinates of its disk centers in some Euclidean model of the hyperbolic plane.

\begin{restatable}{theorem}{thmSeparator}
  \label{thm:separator}
  Let $G$ be a hyperbolic uniform disk graph with radius $r$.  Then
  $G$ has a separator $S$ that can be covered with
  $\Oh\big(\log n \cdot \big(1 + \frac{1}{r}\big)\big)$ cliques, such that all connected components of $G-S$ have at most $\frac{2}{3}n$ vertices. The separator can be computed in $\Oh(n \log n)$ time.
\end{restatable}

Our separator is a carefully chosen line; the proof bounds the number of cliques intersecting the separator via decomposing some neighborhood of the chosen line into a small collection of small diameter regions. Notice that the theorem guarantees a separator with $\Oh(\log n)$ cliques in the hyperbolic setting, and deteriorates linearly as a function of $1/r$. In the almost-Euclidean setting of $r=1/\sqrt{n}$ the separator has $\Oh(\sqrt{n}\log n)$ cliques, which almost matches the $\Oh(\sqrt{n})$ clique  separator for unit disk graphs~\cite{BergBKMZ20}. This is also a direct generalization of the $2$-dimensional separator of Kisfaludi-Bak~\cite{Kisfaludi-Bak20} from $r\in \Theta(1)$ to general radii.

Moreover, 
we observe that the techniques of~\cite{BergBKMZ20,Kisfaludi-Bak20} can be utilized as the above separator implies $\cP$-flattened treewidth $\Oh\left(\log^2 n\left(1+\frac{1}{r}\right)\right)$ for any natural clique partition~$\cP$.

\begin{restatable}{corollary}{corIndepSetAlgo}\label{cor:indepset_algo}
  Let $G$ be a hyperbolic uniform disk graph with radius $r$.
  Then a maximum independent set of $G$ can be
  computed in $n^{\Oh((1+1/r)\log^2 n)}$ time.
\end{restatable}

Thus, a quasi-polynomial $n^{\Oh(\log^2 n)}$ algorithm is possible for all $r\in \Omega(1)$. If $r\leq 1$ --- or even $r\in \Oh(1)$ --- then for any $G\in \HUDG(r)$ the neighborhood of a disk can be covered by a constant number of cliques of $G$ (see also~\cite{Kisfaludi-Bak20} for the case $r \in \Theta(1)$).  Together with our separator in Theorem~\ref{thm:separator}, the machinery of~\cite{BergBKMZ20,Kisfaludi-Bak20} now yields the following result.

\begin{restatable}{corollary}{corAlgos}\label{cor:algos}
Let $G$ be a hyperbolic uniform disk graph with radius $r\in \Oh(1)$.  Then
\textsc{Dominating Set, Vertex Cover, Feedback Vertex Set, Connected
Dominating Set, Connected Vertex Cover, Connected Feedback Vertex Set}
can be solved in $n^{\Oh((1/r)\log n)}$ time, and \textsc{$q$-Coloring}
for $q\in \Oh(1)$, \textsc{Hamiltonian Path} and \textsc{Hamiltonian Cycle}
can be solved in  $n^{\Oh(1/r)}$ time.
\end{restatable}

While the above separator is already powerful and can be used as a basis for quasi-polynomial divide-and-conquer algorithms, the separator size stops improving when increasing the radius $r$ beyond a constant. Nonetheless, a more powerful result is possible for super-constant radius $r$, which requires a different type of separator.  Rather than separating hyperbolic uniform disk graphs, our second structural result is about separating the Delaunay complex, which is the dual of the Voronoi diagram, of a set of sites with pairwise distance at least $2r$.  Note that in some graph $G\in \HUDG(r)$, the disk centers corresponding to an independent set will have pairwise distance more than $2r$, which is what motivated us to study such point sets.

A plane graph (a planar graph with a fixed planar embedding) is called \emph{outerplanar} or $1$-\emph{outerplanar} if all of its vertices are on the unbounded face. A plane graph is called $k$-outerplanar if deleting the vertices of the unbounded face (and all edges incident to them) yields a $(k-1)$-outerplanar graph.
Our result concerns the outerplanarity of a Delaunay complex. We prove the following tight result on the outerplanarity of the Delaunay complex $\cD(S)$ of such a site set $S$. We believe that this result 
may be of independent interest.

\begin{restatable}{theorem}{thmOuterplanarity}
\label{thm:treewidth-of-voronoi-diagram}
  Let $S$ be a set of $n$ sites (points) in $\Hyp^2$ with pairwise distance at least $2r$. Then the Delaunay complex $\cD(S)$ is $1+\Oh\big(\frac{\log n}{r}\big)$-outerplanar.
\end{restatable}

In plane graphs, the treewidth, the outerplanarity, and the dual graph's treewidth and outerplanarity are all within a constant multiplicative factor of each other~\cite{Bodlaender98,BouchitteMT03}.
Consequently, we get the bound $\Oh\big(1+\frac{\log n}{r}\big)$ for the treewidth and outerplanarity of the Voronoi diagram and Delaunay complex of such point sets.  Note that this implies constant outerplanarity for the firmly hyperbolic setting of $r \in \Omega(\log n)$.

\begin{figure}[tbp]
\centering
\includegraphics[]{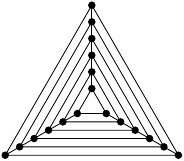}
\caption{A Euclidean Delaunay complex of $n$ points with outerplanarity $n/3$.}\label{fig:Euclouterplan}
\end{figure}

In particular, the theorem shows that for some constant $c$ and $r>c\log n$ the Delaunay complex is outerplanar. We recover the outerplanarity bound of~\cite{Kisfaludi-Bak20} for $r=\Theta(1)$, and get a sublinear outerplanarity bound ($\Oh(\sqrt{n}\log n)$) even in the almost-Euclidean setting of $r=1/\sqrt{n}$. This is surprising in light of the fact that Delaunay-triangulations can have outerplanarity $\Omega(n)$ in the Euclidean setting, as demonstrated by the set of sites in Figure~\ref{fig:Euclouterplan}. Notice however that the construction requires a point set where the ratio of the maximum and minimum distance among the points is $\Omega(n)$; mimicking this construction in $\Hyp^2$ with a point set of minimum distance $r=1/\sqrt{n}$ is not possible, as at distance $n\cdot r=\sqrt{n}$ the hyperbolic distortion is very significant compared to the Euclidean setting.

Using the outerplanarity bound, we are able to give the following algorithm for \textsc{Independent Set}.

\begin{restatable}{theorem}{thmExactalgorithm}
\label{thm:exact}
  Let $G$ be a hyperbolic uniform disk graph with radius $r$ and let $k\geq 0$.
  Then we can decide if there is an independent set of size $k$ in $G$ in
  $\min\big\{n^{\Oh\big(1+\frac{1}{r}\log k\big)},2^{\Oh(\sqrt{n})},n^{\Oh(\sqrt{k})}\big\}$
  time.
\end{restatable}

Our algorithm yields the first component of the running time
($n^{\Oh\big(1+\frac{1}{r}\log k\big)}$), but for very small $r$ we can
switch to the general disk graph algorithms with running times
$2^{\Oh(\sqrt{n})}$ by the authors in~\cite{BergBKMZ20} and
$n^{\Oh(\sqrt{k})}$ by Marx and Pilipczuk~\cite{MarxP22}.
Our algorithm is best possible under ETH
for $r=\Theta(1)$ as~\cite{Kisfaludi-Bak20} showed a lower bound of
$n^{\Omega(\log n)}$ (for large $k$). Notice moreover that the almost-Euclidean case of $r=1/\sqrt{n}$
gives a running time of $n^{\Oh(\sqrt{n}\log k)}$, which almost
recovers the Euclidean running time of $2^{\Oh(\sqrt{n})}$. Recall
that the Euclidean running time is also ETH-tight for Euclidean unit
disks~\cite{BergBKMZ20}. While this Euclidean lower bound cannot be
directly applied to $\HUDG(1/\sqrt{n})$, it can be adapted to the
this setting, so the running time lower bound~$2^{\Omega(\sqrt{n})}$
holds also in the hyperbolic
plane (assuming ETH). See~\cite{Kisfaludi-Bak21} for a similar adaptation of a
Euclidean lower bound to the setting of $r= 1/\sqrt{n}$.

Finally, but most surprisingly, we get a polynomial running time for \textsc{Independent Set} in the firmly hyperbolic setting of $r\in\Omega(\log n)$; in fact, our algorithm is polynomial already for $r\in\Omega(\log k)$. In particular, this provides a polynomial exact algorithm for hyperbolic random graphs.  It thereby directly generalizes the algorithm of~\cite{BlasiusFFK23} which gives a polynomial running time in hyperbolic random graphs with high probability: our algorithm has no assumptions on the input distribution and provides a polynomial worst-case running time.

The underlying idea of our exact algorithm can be regarded as a dynamic programming algorithm along the (unknown) tree decomposition of the Voronoi diagram of the solution disks, inspired by Marx and Pilipczuk~\cite{MarxP22}. More precisely, both~\cite{MarxP22} and the present paper use so-called \emph{sphere cut decompositions}, a variant of branch decompositions for plane graphs~\cite{Dorn_Penninkx_Bodlaender_Fomin_2010} to guide the algorithm. However, due to our setting, a simple adaptation of the divide-and-conquer approach of~\cite{MarxP22} does not give the desired running time for us (it loses another logarithmic factor in the exponent). To get around this problem, we extend Marx and Pilipczuk's technique of noose-based separators to a noose-based dynamic programming algorithm. Additionally, we introduce \emph{well-spaced} and \emph{valid} sphere cut decompositions, and prove that they are in one-to-one correspondence with independent sets of $G$.

Finally, we consider approximation algorithms for \textsc{Maximum Independent Set} when the graph in question has low ply, that is, each point of $\Hyp^2$ is covered by at most $\ell$ disks.  We show the following.

\begin{restatable}{theorem}{thmApprox}
  \label{thm:approx-independent-set}
  Let $\eps\in (0,1)$ and let $G$ be a hyperbolic uniform disk graph of radius $r$ and ply $\ell$.
  Then a $(1-\eps)$-approximate maximum independent set of $G$ can be found in
  $\Oh\left( n^4 \log n \right) + n \cdot \left( \frac{\ell}{\eps} \right)^{\Oh(1 + \frac{1}{r} \log\frac{\ell}{\eps})}$ time.
\end{restatable}

An important component of the algorithm's correctness proof is bounding the so-called \emph{degeneracy} of hyperbolic uniform disk graphs in terms of their clique size and in terms of their ply. This generalizes the Euclidean results of~\cite{MaratheBHRR95}.
When compared to the $n^{\Oh(1/\eps)}$ algorithm for disk graphs by Chan~\cite{Chan03} or the $2^{\Oh(1/\eps)}n$ algorithm in planar graphs by Baker~\cite{Baker94}, both of which are conditionally optimal~\cite{Marx07a}, our algorithm has a surprising quasi-polynomial dependence on $1/\eps$ rather than exponential. The dependence on $r$ exhibits the classic Euclidean-to-hyperbolic scaling seen in our earlier structural results. For $\eps=1/n$, we are guaranteed to get the optimal solution, and setting $\ell=n$ covers the general case. The resulting running time for $1/\eps=\ell=n$ is $n^{\Oh\left(1+\frac{1}{r}\log n\right)}$, which matches our exact algorithm. Consequently, our approximation scheme is a generalization of our exact algorithm.

\subparagraph{Outline.}
The remainder of this article is structured as follows.  After the
introduction of important definitions and notation in
\cref{sec:prelim}, we first give a high-level overview of our results
in \cref{sec:overview}.  The subsequent sections then discuss all of
our results in more detail.  Finally, we conclude with a short summary
and discussion of future directions in \cref{sec:conclusion}.

\section{Preliminaries}\label{sec:prelim}

We introduce fundamental definitions and notation used throughout the
paper.  Additional concepts and notation are introduced in their
relevant sections as needed.

\subparagraph*{Graphs and treewidth.}

A \emph{graph} $G = (V, E)$ is a tuple of $n$ \emph{vertices} $V$ and
$m$ \emph{edges} $E \subseteq V \times V$.  We sometimes write $V(G)$
and $E(G)$ to make explicit which graph we refer to.  Unless mentioned
otherwise, we only consider \emph{simple} graphs, i.e., there are no
multi edges and no loops.  A graph $H$ is a \emph{subgraph} of $G$ if
$V(H) \subseteq V(G)$ and $E(H) \subseteq E(G)$.  The subgraph $H$ is
\emph{induced} by $V(H)$ if $u, v \in V(H)$ and $\{u, v\} \in E(G)$
implies that $\{u, v\} \in E(H)$.  Two vertices are \emph{adjacent} if
they appear together in the same edge and an edge is \emph{incident}
to its vertices. The \emph{degree} of a vertex is the number of
incident edges.  A \emph{cycle} is a connected graph in which every
vertex has degree~$2$.

An \emph{independent set} in a graph is a set of pairwise non-adjacent
vertices. In the \textsc{Independent Set} problem, we are
given\footnote{In this article, the graphs are always intersection
  graphs given through their geometric representation.} a graph $G$
and a number $k$, and the goal is to decide if there is an independent
set of size $k$. In \textsc{Maximum Independent Set}, the input is a
graph $G$, and the goal is to output the maximum value $k$ such that
$G$ has an independent set of size $k$.

Let $G = (V, E)$ be a graph. A vertex set $S \subseteq V$ is a
\emph{separator} if $V$ can be partitioned into non-empty subsets $S$,
$V_1$, and $V_2$ such that there is no edge between $V_1$ and $V_2$,
i.e., for every $v_1 \in V_1$ and $v_2 \in V_2$, it holds that
$\{v_1, v_2\} \notin E$.  The separator is $\beta$-\emph{balanced} if
$|V_1|, |V_2| \le \beta n$ for some $\beta<1$.

A \emph{tree decomposition} of a graph $G=(V,E)$ is a pair $(T,\sig)$
where $T$ is a tree and $\sig$ is a mapping from the vertices of $T$
to subsets of $V$ called \emph{bags}, with the following
properties. Let $\bags(T,\sig) := \{ \sig(u): u \in V(T) \}$ be the
set of bags associated to the vertices of $T$. Then we have: (1) For
any vertex $u\in V$ there is at least one bag in $\bags(T,\sig)$
containing it. (2) For any edge $(u,v)\in E$ there is at least one bag
in $\bags(T,\sig)$ containing both $u$ and $v$. (3) For any vertex
$u\in V$ the collection of bags in $\bags(T,\sig)$ containing $u$
forms a subtree of~$T$.  The \emph{width} of a tree decomposition is
the size of its largest bag minus~1, and the \emph{treewidth} of a
graph $G$ equals the minimum width of a tree decomposition of $G$.

A \emph{clique-partition} of $G$ is a partition $\cP$ of $V(G)$ where
each partition class $C\in \cP$ forms a clique in $G$.  The
\emph{$\cP$-contraction} of $G$ is the graph obtained by contracting
all edges induced by each partition class, and removing parallel
edges; it is denoted by $G_\cP$.  The \emph{weight} of a partition
class (clique) $C\in \cP$ is defined as $\log(|C|+1)$. Given a set
$S\subset \cP$, its weight $\gamma(S)$ is defined as the sum of the
class weights within, i.e.,
$\gamma(S):= \sum_{C\in S} \log(|C|+1)$. Note that the weights of the
partition classes define vertex weights in the contracted graph
$G_\cP$.

We will need the notion of {\em weighted
  treewidth}~\cite{EijkhofBK07}.  Here each vertex has a weight, and
the {\em weighted width} of a tree decomposition is the maximum over
the bags of the sum of the weights of the vertices in the bag (note:
without the minus $1$).  The {\em weighted treewidth} of a graph is
the minimum weighted width over its tree decompositions.  Now let
$\cP$ be a clique partition of a given graph $G$.  We apply the
concept of weighted treewidth to $G_{\cP}$, where we assign each
vertex $C$ of $G_\cP$ the weight~$\log(|C|+1)$, and refer to this
weighting whenever we talk about the weighted treewidth of a
contraction $G_\cP$. For any given $\cP$, the weighted treewidth of
$G_\cP$ with the above weighting is referred to as the
\emph{$\cP$-flattened treewidth} of $G$.

\subparagraph*{Planarity.}

A \emph{drawing} $\Gamma$ of $G$ maps its vertices to different points
in $\mathbb R^2$ and its edges to curves between its endpoints, i.e.,
$\Gamma(v) \in \mathbb R^2$ and $\Gamma(\{u, v\})$ is a curve from
$\Gamma(u)$ to $\Gamma(v)$.  We sometimes also identify a vertex with
its position, i.e., $v$ is used to refer to the vertex as well as to
the point $\Gamma(v)$.  A drawing is \emph{planar} if no two edges
intersect except at a common endpoint and the graph $G$ is planar if
it has a planar drawing.  Two planar drawings $\Gamma_1$ and
$\Gamma_2$ are \emph{combinatorially equivalent} if there is a
homeomorphism of the plane onto itself that maps $\Gamma_1$ to
$\Gamma_2$.  The equivalence classes with respect to this equivalence
relation are called \emph{(planar) embeddings} and the drawings within
such a class are said to \emph{realize} the embedding.  A planar graph
together with a fixed embedding is also called a \emph{plane} graph.

Consider a graph $G$ with a planar drawing $\Gamma$.  Removing all
edges (i.e., their drawing) from $\mathbb R^2$ potentially disconnects
$\mathbb R^2$ into several connected components, which are called
\emph{faces}.  Exactly one of these faces is unbounded.  It is called
the \emph{outer face}; all other face are \emph{inner faces}.  The
boundary of each face is a cyclic sequence of vertices and edges.  For
different drawings realizing the same embedding, these sequences are
the same.  Thus, to talk about faces and their boundaries, we do not
require a drawing but only an embedding.
Vertices and edges on this boundary are \emph{incident} to the face
and if a vertex appears multiple times on the boundary, it has
multiple \emph{incidences} to the face.  An edge can also have two
incidences to the same face, which is the case if and only if it is a
\emph{bridge}, i.e., removing it separates the graph in two
components.
In a plane graph, a vertex incident to the outer face is called
\emph{outer vertex}; other vertices are \emph{inner vertices}.

Given a plane graph $G$, the \emph{dual} $G^\star$ has one vertex for
each face of $G$.  For every edge $e$ in $G$, the dual $G^\star$
contains the dual edge $e^\star$ that connects the two faces incident
to $e$ (which results in a loop, if $e$ is a bridge).  The dual
$G^\star$ is itself a plane graph and its dual is
$G^{\star\star} = G$.  The \emph{weak dual} of $G$ is $G^\star$
without the vertex representing the outer face of $G$.

A \emph{polygon} is a cycle together with a planar drawing in which
all edges are drawn as line segments. Note that this enables the use
of graph notation, e.g., using $V(P)$ to refer to the vertices of
polygon $P$.

\subparagraph*{Hyperbolic geometry.}

\begin{figure}[tbp]
	\centering
	\begin{minipage}[b]{0.31\textwidth}
		\centering
		\includegraphics[page=1]{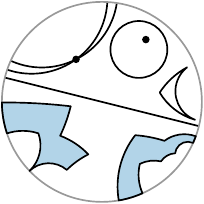}
		\subcaption{}
		\label{fig:prelim:poincare}
	\end{minipage}
	~ 
	\begin{minipage}[b]{0.31\textwidth}
		\centering
		\includegraphics[page=2]{figures/poincare_klein.pdf}
		\subcaption{}
		\label{fig:prelim:klein}
	\end{minipage}
	~ 
	\begin{minipage}[b]{0.31\textwidth}
		\centering
		\includegraphics[page=1]{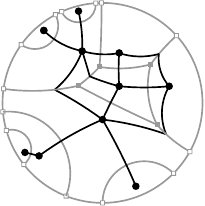}
		\subcaption{}
		\label{fig:prelim:delaunay}
	\end{minipage}
	\caption{Part (a) shows a visualization of the Poincaré disk
		including (parallel) lines, a circle with its center, a triangle
		and generalized polygons.  Part (b) shows a visualization of the
		Beltrami-Klein model including (parallel) lines, a triangle, and a
		generalized polygon. Part (c) shows a set of sites $S \subseteq \Hyp^2$
		in the Poincaré disk model, $\mathcal{D}(S)$ (black), and
		$\mathcal{V}(S)$ (gray); Voronoi vertices are filled squares,
		ideal Voronoi vertices are empty squares.}%
	\label{fig:prelim}
\end{figure}

We write $\Hyp^2$ to refer to the hyperbolic plane and $\mathbb{R}^2$
to refer to the Euclidean plane.  The \emph{Poincaré disk} is a model
of $\Hyp^2$ that maps the whole hyperbolic plane into the interior of
a Euclidean unit disk; see Figure~\ref{fig:prelim:poincare}.  We refer to
the center of the Poincaré disk as \emph{origin}.
Straight lines in $\Hyp^2$ are represented as circular arcs
perpendicular to the boundary of the Poincaré disk or as chords
through the origin.  Hyperbolic circles are represented as Euclidean
circles.  The center of the hyperbolic circle is, however, farther
from the origin than the Euclidean center of its representation.  The
Poincaré disk is \emph{conformal}, i.e., angle preserving.

Points on the boundary of the Poincaré disk are called \emph{ideal}
points.  These are not part of the hyperbolic plane, but form a
boundary of infinitely far points.  An \emph{ideal arc} between two
ideal points $q_1$ and $q_2$ is the set of ideal points between $q_1$
and $q_2$ when moving clockwise on the boundary of the Poincaré disk.
A \emph{generalized polygon} is a cycle together with a planar drawing
that maps each vertex to a point in $\mathbb H^2$ or to an ideal
point.  Each edge is either a line segments between points, a ray from
a point to an ideal point, a line between two ideal points, or an
ideal arc between two ideal points.  Note that the vertices do not
completely determine the generalized polygon, as two ideal points can
be connected via a line or an ideal arc.  Moreover, a 2-cycle can
yield a generalized polygon by mapping both vertices to ideal points
and one edge to the line and the other to the ideal arc between the
two.  Figure~\ref{fig:prelim:poincare} shows two generalized polygons
with their interior shaded in blue.

For a straight line $\ell$ and a point $p \notin \ell$, there are
infinitely many \emph{parallel} lines through $p$, i.e., lines through
$p$ that do not intersect $\ell$.  Two of these parallel lines are
special in the sense that they are the closest to not being parallel.
They each share an ideal endpoint with $\ell$ and are called
\emph{limiting parallels}.  Let $q \in \ell$ be such that $pq$ is
perpendicular to $\ell$.  Then the angle between $pq$ and the two
limiting parallels is the same on both sides and only depends on the
length $x = |pq|$.  It is called the \emph{angle of parallelism} and
denoted by $\Pi(x)$.  It holds that
$\sin(\Pi(x)) = 1 /
\cosh(x)$~\cite[page~402]{Euclid_Non_Euclid_Geomet-Green93}.

Let $a$ be the sum of interior angles of a hyperbolic triangle.  Then
$a < \pi$ and the area of the triangle is $\pi - a$.  Note that this
implies that the area of a triangle is upper bounded by $\pi$.  More
generally, the area of a $k$-gon is bounded by $(k - 2) \pi$.  The
area of a disk with hyperbolic radius $r$ is $4π \sinh^2(r/2)$, which
is in $\Theta(e^r)$ for $r \to \infty$ and in $\Theta(r^2)$ for $r \to 0$.

The \emph{Beltrami--Klein model} also uses the interior of the
Euclidean unit disk as ground space.  Hyperbolic straight lines are
represented as chords of the unit disk, see Figure~\ref{fig:prelim:klein}.
This model is not conformal, but enables an easy translation of
Euclidean line arrangements into the hyperbolic plane.

\subparagraph*{Delaunay complexes and Voronoi diagrams.}

Let $S$ be a set of \emph{sites}, i.e., a set of designated points in
$\Hyp^2$.  The \emph{Voronoi cell} of a site $s \in S$ is the set of
points that are closer to $s$ than to any other site.  When considered
in the Poincaré disk, the boundary of each Voronoi cell is a
generalized polygon\footnote{We note that Voronoi cells containing
	ideal points in their boundary are actually unbounded in the
	hyperbolic plane.}.  The non-ideal segments are part of the
perpendicular bisector of two sites and are called \emph{Voronoi
	edges}; they are illustrated in blue in
Figure~\ref{fig:prelim:delaunay}.  The points in which three or more
Voronoi edges meet are called \emph{Voronoi vertices}.  The ideal
points in which unbounded Voronoi edges end are called \emph{ideal
	Voronoi vertices}.  We define the \emph{Voronoi diagram}
$\mathcal V(S)$ of $S$ as the following plane graph.  Its vertex set
is the set of all (ideal) Voronoi vertices.  The edge set of
$\mathcal V(S)$ is comprised of the Voronoi edges and the set of ideal
arcs connecting consecutive ideal vertices (red in
Figure~\ref{fig:prelim:delaunay}).

The weak dual of the Voronoi diagram is called \emph{Delaunay complex}
and denoted by $\mathcal D(S)$; it is illustrated in black in
Figure~\ref{fig:prelim:delaunay}.  The edges of $\mathcal D(S)$ are
exactly the edges dual to the Voronoi edges, i.e., the edges of
$\mathcal V(S)$ that are not ideal arcs.  Thus, sites $s_1 \in S$ and
$s_2 \in S$ are connected in the $\mathcal D(S)$ if and only if their
Voronoi cells share a boundary.  In contrast to the Euclidean case,
the outer face of $\mathcal D(S)$ is not the convex hull of $S$.  In
fact, its boundary is not necessarily a simple cycle; see
Figure~\ref{fig:prelim:delaunay}.  A site is an outer vertex of
$\mathcal D(S)$ if and only if the corresponding Voronoi cell is
unbounded.

Let $f_\infty$ denote the outer face of the Delaunay complex
$\mathcal D(S)$.  For each incidence of an edge $e$ of $\mathcal D(S)$
to $f_\infty$, the dual edge $e^\star$ has one ideal vertex as
endpoint.  If $e$ is not a bridge, $e^\star$ is a ray connecting a
Voronoi vertex with an ideal Voronoi vertex.  If $e$ is a bridge
$e^\star$ is a line connecting two ideal Voronoi vertices; one for
each incidence of $e$ to $f_\infty$.

The Voronoi cells are convex.  Thus, drawing each edge $\{u, v\}$ of
the Delaunay complex $\mathcal D(S)$ by choosing an internal point $p$
on the Voronoi edge dual to $\{u, v\}$ and then connecting $u$ via $p$
to $v$ using the two line segments $up$ and $pv$ yields a planar
drawing of $\mathcal D(S)$; see Figure~\ref{fig:prelim:delaunay}.  It
has the property that each edge of $\mathcal D(S)$ intersects its dual
Voronoi edge in exactly one point and it intersects no other Voronoi
edges.

\section{Overview of our techniques}\label{sec:overview}

\subsection{Balanced separators in HUDGs}

Let $G$ be a hyperbolic uniform disk graph with $n$ vertices and
radius $r$.  We show that $G$ has a balanced separator that can be
covered by $\Oh((1 + 1 / r)\log n)$ cliques; see
Theorem~\ref{thm:separator}.  The overall argument is as follows.  We
find a double wedge bounded by two lines $\ell_1$ and $\ell_2$ that
contains no vertex in its interior and separates the other vertices of
$G$ in a balanced fashion, which is ensured by selecting the center of the double wedge to be the hyperbolic \emph{centerpoint}~\cite{Kisfaludi-Bak21}. In Figure~\ref{fig:separator-hypercycle},
the double wedge with apex $p$ is shaded gray and contains no vertices
and the regions above and below the wedge each contain a constant
fraction of the vertices.

Let $m$ be the angular bisector of the wedge.  As separator, we use
the set of all vertices that have distance at most $r$ from $m$.  The
set of points with distance exactly $r$ from $m$ forms two curves that
are called \emph{hypercycles} with axis $m$.  Thus, vertices belong to
the separator if they lie on or between the two hypercycles.  The
hypercycles are shown in Figure~\ref{fig:separator-hypercycle} and the
region where separator vertices can lie is shaded blue.  In the
Poincaré disk, hypercycles are arcs of Euclidean circles that meet the
boundary of the disk at the same ideal points as their axis but at a
non-right angle.  Observe that any vertex from above the top
hypercycle has distance greater than $2r$ to any vertex below the
bottom hypercycle.  Thus, the vertices in the blue region indeed form
a balanced separator.

It remains to show that the graph induced by vertices in the blue
region can be covered with few cliques.  For this, we cover the blue
region with boxes as shown in Figure~\ref{fig:separator-boxes}.  We
show that each of these boxes has diameter at most $2r$, implying that
the vertices within one box induce a clique.  Moreover, we show that
we only need $\Oh((1 + 1/r) \log n)$ such boxes.  For the latter, we
need a lower bound on the opening angle of the empty double wedge.
Observe that a larger opening angle, $\varphi$, in
Figure~\ref{fig:separator-triangle}, shrinks the blue region, which
reduces the number of boxes required to cover it.

In Section~\ref{sec:cover-separ-with-cliques}, we give a bound on the
number of cliques required to cover the separator, parameterized by
the opening angle of the empty double wedge.  In
Section~\ref{sec:balanced-separators} we show that there is a wedge
with not too small opening angle that yields a balanced separator.
The proof is constructive, i.e., given the graph with vertex
positions, we can efficiently find the wedge and thereby the balanced
separator of~\Cref{thm:separator}.

\begin{figure}[tbp]
  \centering
  \begin{minipage}[b]{0.48\textwidth}
    \centering
    \includegraphics[page=2]{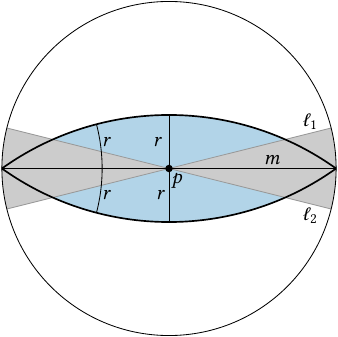}
    \subcaption{\label{fig:separator-hypercycle}}
  \end{minipage}
  \hfill
  \begin{minipage}[b]{0.48\textwidth}
    \begin{minipage}[b]{\textwidth}
      \centering
      \includegraphics[page=3]{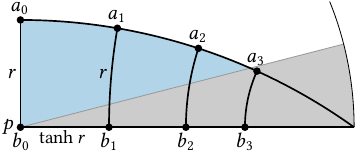}
      \subcaption{\label{fig:separator-boxes}}
    \end{minipage}\\[\baselineskip]
    \begin{minipage}[b]{\textwidth}
      \centering
      \includegraphics[page=4]{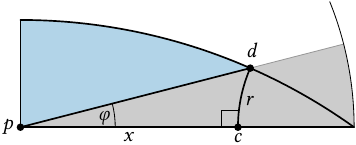}
      \subcaption{\label{fig:separator-triangle}}
    \end{minipage}
  \end{minipage}
  \caption{Illustration of our separator using the Poincaré disk
    model.  (\subref{fig:separator-hypercycle}) The axis $m$ is chosen such that it separates the vertices in a balanced fashion.  The double wedge between  $\ell_1$ and $\ell_2$ contains no vertices.  The separator consists of the blue set of all vertices of distance at most $r$ to the axis $m$.  (\subref{fig:separator-boxes}) We cover the separator with boxes that have diameter $2r$ and thus form cliques.  (\subref{fig:separator-triangle}) If the opening angle $\varphi$ is sufficiently large, then $x$ is sufficiently small and we only need few boxes to cover the whole separator.}
  \label{fig:separator}
\end{figure}

\subsection{Outerplanarity of hyperbolic Delaunay complexes}

Let $S$ be a set of $n$ sites in $\Hyp^2$ with pairwise distance at
least $2r$.  Note that $S$ interpreted as a hyperbolic uniform disk
graph with radius $r$ forms an independent set.  To prove
Theorem~\ref{thm:treewidth-of-voronoi-diagram}, giving an upper bound
on the outerplanarity of the Delaunay complex $\mathcal D(S)$, we use
two different types of arguments.  The first argument is based on the
observation that sufficiently large radius implies that any inner
vertex of $\mathcal D(S)$ requires many other sites around it to
shield it from being an outer vertex.  This is similar to the
observation in the introduction (recall
Figure~\ref{fig:impact-of-radius}) that for high radius $r$, large
stars can be realized.  The argument then goes roughly as follows: If
all inner vertices have degree more than $6$, then we need a linear
fraction of vertices on the outer face to account for the fact that
the average degree in planar graphs is bounded by $6$.  This already
implies that iteratively deleting the vertices on the outer face takes
at most $\log n$ steps as each step deletes a constant fraction of the
vertices.  A larger radius $r$ leads to higher degrees of inner vertices, which in turn yields stronger bounds for the outerplanarity.

This type of argument can only work if the radius $r$ is sufficiently
large.  For smaller radii, the second type of argument considers
layers of bounded Voronoi cells (which correspond to inner vertices of
$\mathcal D(S)$) around one fixed center cell.  As the union of these
layers is bounded by a polygon with a linear number of vertices, its
area is linear.  In contrast to that, we will see that the area of the
layers grows exponentially with the layer (with the base depending on
$r$), showing that there cannot be too many layers.

In the following, we give a sketch of how the second argument based on
area works for small radii.  Although this is the argument of choice
for radii that are small constants or even decreasing with $n$, we
first give a simple area-based argument for why $r \geq \log n$
implies that the Delaunay complex has no inner vertices.  This gives a
good intuition how the area behaves in the hyperbolic plane and how
this can help us in proving outerplanarity, even though our argument
for small radii is more involved.  Let $s \in S$ be a site.  If $s$ is
an inner vertex of $\mathcal D(S)$, then its Voronoi cell is bounded
and its boundary is a polygon of at most $n - 1$ vertices.  Thus, the
area of this Voronoi cell is less than $(n - 3) \pi$.  Simultaneously,
the disk of radius $r$ around $s$ is included in the Voronoi cell of
$s$ as all sites have pairwise distance at least $2r$.  The area of
this disk is $4\pi\sinh^2(\frac{r}{2})$ which is larger than $e^r$ for
sufficiently large $r$.  From this it follows that bounded cells can
only exist when $r < \log n$.

To make an argument that works for smaller $r$, let $s \in S$ be an
arbitrary but fixed vertex of the Delaunay complex $\mathcal D(S)$.
We partition the vertices of $\mathcal D(S)$ into \emph{layers} by hop
distance from $s$ in $\mathcal D(S)$, i.e., $V_\ell$ is the set of
vertices with distance $\ell$ from $s$.  Let $L$ be the largest
integer such that for all $\ell \le L$ the layer $V_\ell$ contains
only inner vertices.  Note that our goal is to prove an upper bound on
$L$, as this bounds the distance of $s$ to the outer face.

As the Delaunay complex is an internally triangulated plane graph,
each layer $V_\ell$ induces a cycle.  We denote the vertices in
$V_\ell$ with $v_\ell^i$ and number the vertices modulo $|V_\ell|$,
such that vertex $v_\ell^i$ is adjacent to vertices $v_\ell^{i-1}$ and
$v_\ell^{i+1}$; see
Figure~\ref{fig:outerplanarity_small_radius_area_overview}.  Moreover,
we assume the Delaunay graph to be drawn as follows.  For every edge
$\{v_\ell^i, v_\ell^{i + 1}\}$ we choose a point on its dual edge,
i.e., the boundary between the two corresponding Voronoi cells, denote
it with $v_\ell^{i + 0.5}$, and connect $v_\ell^{i}$ and
$v_\ell^{i + 1}$ with two line segments via $v_\ell^{i + 0.5}$.  We
denote the resulting polygon with $P_\ell$ and call it a \emph{layer
  polygon} of $s$.  Note that this yields a sequence of nested
polygons where $P_{\ell - 1}$ lies in the interior of $P_\ell$.

\begin{figure}[tbp]
  \centering
  \begin{subfigure}[b]{0.32\textwidth}
    \centering
    \includegraphics[page=1]{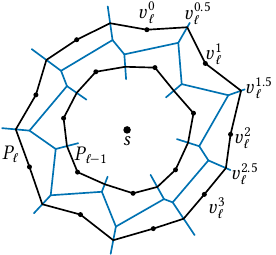}
    \subcaption{\label{fig:outerplanarity_small_radius_area_overview}}
  \end{subfigure}
  \hfill
  \begin{subfigure}[b]{0.32\textwidth}
    \centering
    \includegraphics[page=2]{figures/outerplanarity_small_radius_area}
    \subcaption{\label{fig:outerplanarity_small_radius_area_polygon}}
  \end{subfigure}
  \hfill
  \begin{subfigure}[b]{0.32\textwidth}
    \centering
    \includegraphics[page=3]{figures/outerplanarity_small_radius_area}
    \subcaption{\label{fig:outerplanarity_small_radius_area_layer}}
  \end{subfigure}
  \caption{Illustration of the small-radius case.  (a) shows two
    consecutive layer polygons $P_{\ell - 1}$ and $P_\ell$ in the
    Delaunay complex around $s$.  The Voronoi diagram is shown in
    blue.  We show an upper bound of the area inside $P_\ell$ by
    covering it with the red triangles (b).  We show a lower bound for
    the area between $P_{\ell - 1}$ to $P_\ell$ via the red triangles
    in (c).  Relating these two bounds yields an exponential area
    growth with a basis depending on $r$.}
  \label{fig:outerplanarity_small_radius_area}
\end{figure}

Our goal then is to show that the area inside $P_\ell$, denoted by
$\area(P_\ell)$, grows exponentially in $\ell$ with a basis $b(r)$
depending on $r$.  As $\area(P_L)$ for the outermost layer $L$ is
upper bounded by something linear in $n$, there cannot be too many
layers.  The interesting part is proving the exponential growth.  For
this, we show that the area gain $\area(P_\ell) - \area(P_{\ell - 1})$
in layer $\ell$ makes up at least some sufficiently large fraction of
the area $\area(P_\ell)$.  For this, we give an upper bound for
$\area(P_\ell)$ and relate it to a lower bound for
$\area(P_\ell) - \area(P_{\ell - 1})$.

How we derive these bounds is illustrated in
Figure~\ref{fig:outerplanarity_small_radius_area_polygon} and
Figure~\ref{fig:outerplanarity_small_radius_area_layer}, respectively.
For the upper bound on $\area(P_\ell)$, we cover $P_\ell$ with
triangles connecting every edge of $P_\ell$ with the vertex $s$; see
the two red triangles in
Figure~\ref{fig:outerplanarity_small_radius_area_polygon} for the two
edges $\{v_\ell^{0.5}, v_\ell^{1}\}$ and
$\{v_\ell^{1}, v_\ell^{1.5}\}$ of $P_\ell$.  For the lower bound, we
find a set of disjoint triangles that lie between $P_\ell$ and
$P_{\ell - 1}$.  For this, observe that for every vertex
$v_\ell^i \in V_\ell$ in layer $\ell$, the Voronoi cell of $v_\ell^i$
completely contains the disk of radius $r$ around $v_\ell$ as the
sites have pairwise distance at least $2r$.  Thus, the two triangles
illustrated in red for $v_\ell^1$ in
Figure~\ref{fig:outerplanarity_small_radius_area_layer} satisfy the
property of lying between $P_\ell$ and $P_{\ell - 1}$.  As the two
triangles can in principle intersect, we choose for each vertex in
layer $\ell$ the larger of the two triangles.  Note that this gives a
collection of $|V_\ell|$ disjoint triangles, as each chosen triangle
lies in a different Voronoi cell.  Thus, the total area of these
triangles gives a lower bound for
$\area(P_\ell) - \area(P_{\ell - 1})$.

It then remains to relate the upper bound for $\area(P_\ell)$, i.e.,
the area of the red triangles in
Figure~\ref{fig:outerplanarity_small_radius_area_polygon}, with the
lower bound for $\area(P_\ell) - \area(P_{\ell - 1})$, i.e., the area
of larger red triangle in
Figure~\ref{fig:outerplanarity_small_radius_area_layer}.  Intuitively,
this does not seem too far fetched for the following reasons.  First,
the triangles in
Figure~\ref{fig:outerplanarity_small_radius_area_polygon} share a side
with the triangles in
Figure~\ref{fig:outerplanarity_small_radius_area_layer}.  Secondly,
the other sides of the triangles in
Figure~\ref{fig:outerplanarity_small_radius_area_layer} have length at
least $r$.  Thus, the triangles in
Figure~\ref{fig:outerplanarity_small_radius_area_layer} cannot be too
much smaller than those in
Figure~\ref{fig:outerplanarity_small_radius_area_polygon}.

The details of this area-based argument for small radii can be found
in Section~\ref{sec:small-radius-case}.  Together with the
degree-based bound for large radii, see
Section~\ref{sec:large-radius-case} for details, we
obtain~\Cref{thm:treewidth-of-voronoi-diagram}.

\subsection{Exact algorithm}

Our exact algorithm stated in Theorem~\ref{thm:exact} works roughly as
follows; see Section~\ref{sec:exact_algo_appendix} for details.
Instead of computing the independent set directly, we define some
additional structure for a given independent set.  The algorithm then
finds this additional structure such that the corresponding
independent set is maximized.  The additional structure is illustrated
for two example independent sets in
Figure~\ref{fig:exact_algo_overview}.  In the following, we explain
this from top to bottom.

\begin{figure}[t]
	\centering
	\includegraphics{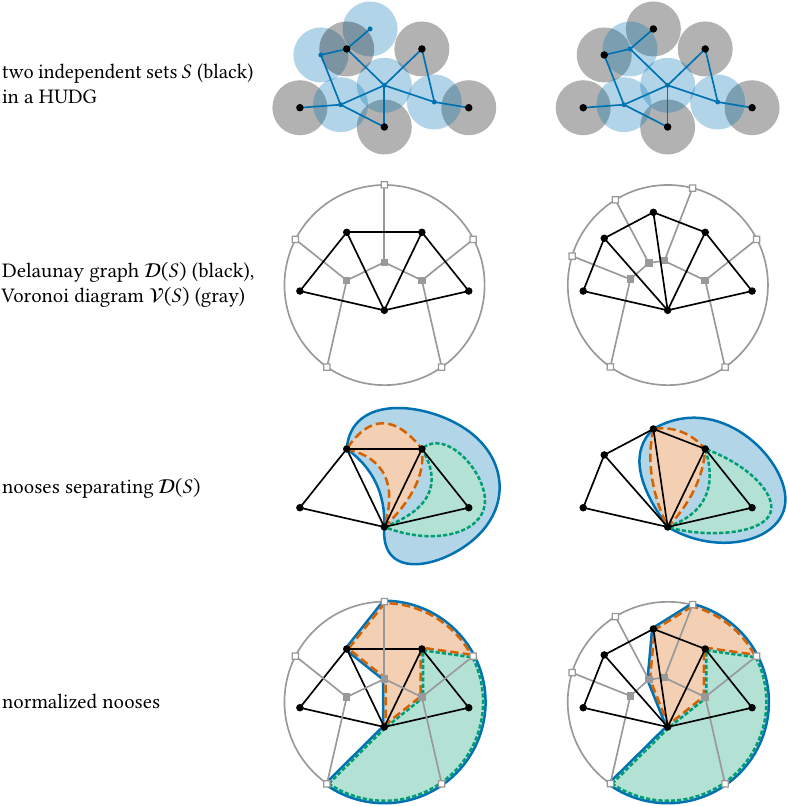}
	\caption{Illustration of our exact algorithm in
		Theorem~\ref{thm:exact}.}
	\label{fig:exact_algo_overview}
\end{figure}

Let $G$ be a hyperbolic uniform disk graph of radius $r$ together with
an independent set $S$ (top row in
Figure~\ref{fig:exact_algo_overview}).  As there are no edges between
vertices in $S$, they have pairwise distance at least $2r$.  Thus, the
Delaunay complex $\mathcal D(S)$ (second row in
Figure~\ref{fig:exact_algo_overview}) has low outerplanarity due to
Theorem~\ref{thm:treewidth-of-voronoi-diagram}.  Low outerplanarity
implies small treewidth, which implies small
branchwidth~\cite{Robertson_Seymour_1991}.  Moreover, as
$\mathcal D(S)$ is a planar graph, we can wish for a so-called
\emph{sphere cut decomposition} of low width, which is a branch
decomposition with some additional
properties~\cite{DBLP:journals/combinatorica/SeymourT94,Dorn_Penninkx_Bodlaender_Fomin_2010,MarxP22}.
In a nutshell, we get a hierarchy of \emph{nooses}, where each noose
is a closed curve that intersects the graph only at vertices and
functions as a vertex separator, splitting interior edges from
exterior ones.  The hierarchy is a binary tree in the sense that every
noose containing more than one edge in its interior has two child
nooses partitioning these edges further.  Row three in
Figure~\ref{fig:exact_algo_overview} shows three nooses of the
Delaunay complex $\mathcal D(S)$.  The solid blue noose is the parent
of the dashed red and green noose.

The width of a sphere cut decomposition is the maximum number of
vertices on a noose.  As $\mathcal D(S)$ has low outerplanarity, we
have a low-width sphere cut decomposition, i.e., each noose goes
through only few vertices.  As a last step, we normalize the nooses by
specifying their exact geometry, keeping their combinatorial structure
intact; see the last row of Figure~\ref{fig:exact_algo_overview}.  At
its core, the normalized nooses are polygons that alternate between
vertices in $S$ and Voronoi vertices, with some special treatment for
the outer face of $\mathcal D(S)$.

Observe that the two examples (left and right column) illustrate
different independent sets $S$, but share the green normalized noose.
This is not a coincidence, but actually happens a lot.  In fact,
despite the potentially exponential number of independent sets, we can
bound the total number of normalized nooses that we can obtain in this
way by $n^{\Oh(1 + \frac{\log n}{r})}$.  This follows from the fact
that each noose visits only $\Oh(1 + \frac{\log n}{r})$ vertices due
to Theorem~\ref{thm:treewidth-of-voronoi-diagram}.  Thus, the
normalized nooses are essentially polygons with
$\Oh(1 + \frac{\log n}{r})$ corners and the set of potential corners
is polynomial, as each corner is a vertex of $G$ or a Voronoi vertex
(which is defined by three vertices of $G$).

Our independent set algorithm works as follows.  We start by listing
all candidates for normalized nooses.  We then use a dynamic program
to combine nooses into a hierarchy of polygons.  We make sure that the
resulting hierarchy is valid in the sense that it corresponds to a
sphere cut decomposition.  Additionally, we enforce that it is
\emph{well-spaced}, which lets us show that all vertices of $G$
visited by the polygons in the hierarchy form an independent set.
Then, finding a maximum independent set is equivalent to maximizing
the number of vertices visited by polygons in this polygon hierarchy,
which can be done with the dynamic program.

\subsection{Approximation algorithm}
The approximation algorithm of Theorem~\ref{thm:approx-independent-set}
(fully described in \Cref{sec:approx_algo_appendix}) works
similarly to Lipton and Tarjan's version for planar graphs~\cite{LiptonT80}:
we repeatedly apply a separator until we get small patches and then solve each
patch individually using the exact algorithm of Theorem~\ref{thm:exact}.
Vertices in the separators get ignored, so we get an $(1-\eps)$-approximation
when these make up only an $\eps$ fraction of the maximum independent set.
We guarantee this using a lower bound of $\Omega(n / \ell)$ on the size of a
maximum independent set, which follows from a proof that hyperbolic uniform
disk graphs of ply $\ell$ always have a vertex of degree at most $4\ell - 1$,
making them $(4\ell - 1)$-degenerate.

\section{Grids, and stars in hyperbolic uniform disk graphs}

In this section we revisit a claim from the introduction and show that
hyperbolic uniform disk graphs with radius $1/n^{3}$ and radius
$\log n$ are incomparable.  The former class contains grids but no
large stars while the latter contains all stars but no large grids.

\begin{theorem}\label{thm:gridandstar}
Let $\Gamma_k$ denote the $k \times k$ grid graph on $n = k^2$ vertices, and let $S_n$ denote the star on $n$ vertices. Then we have the following:
\begin{description}
\item[(i)] $\Gamma_k\in \HUDG(1/n^3)$ for all $k$,
\item[(ii)] $S_n\not\in \HUDG(1/n^3)$ for $n \ge 8$,
\item[(iii)] $\Gamma_k \not\in \HUDG(\log n)$ for $k \ge 5$,
\item[(iv)] $S_n\in \HUDG(\log n)$ for all $n$.
\end{description}
\end{theorem}

\begin{proof}
(i) Consider a mapping from grid vertices to coordinates in the Poincar\'e half-plane model, where the vertex $v=(i,j)$ is mapped to the coordinate $p_v=(2i/n^3,1+2j/n^3)\in\mathbb{R}\times\mathbb{R}_{>0}$ for $(i,j)\in [k]\times[k]$. For a detailed introduction to the upper half-plane model see \cite[p.~128]{ratcliffe_foundations_2019}.
The distance between two points $p_u=(x_u,y_u)$ and $p_v=(x_v,y_v)$ in the upper half-plane can be expressed by $d_{\Hyp^2}(p_u,p_v)=\arcosh(1+\frac{\norm{p_u-p_v}^2}{2y_uy_v})$.

For two grid vertices $u=(i,j)$, $v=(i',j')$, we first consider the
terms $\norm{p_u-p_v}^2$ and $2y_uy_v$ appearing in the distance
function.  We start by giving a lower and upper bound for $2y_uy_v$.
Observe that $y_u, y_v \ge 1$ and thus $2y_uy_v \ge 2$.  Conversely,
as the $j,j' \le \sqrt{n}$, we have $y_u, y_v \le 1 + 2\sqrt{n}/n^3$.
As we can assume $n \ge 4$, we get $y_u, y_v \le 1 + \frac{1}{16}$ and
thus $2y_uy_v \le 2 (1 + \frac{1}{16})^2 < 3$.  Thus,
$2 \le 2y_uy_v < 3$.  For the term $\norm{p_u-p_v}^2$ we get
\begin{equation*}
  \norm{p_u-p_v}^2 = \frac{4(i-i')^2}{n^6} + \frac{4(j-j')^2}{n^6} = \frac{4}{n^6} \left((i - i')^2 + (j - j')^2\right).
\end{equation*}
In case $u$ and $v$ are adjacent, one of $(i-i')^2$ and $(j-j')^2$ is
$1$ and the other is $0$ and thus we get $\norm{p_u-p_v}^2 = 4/n^6$.
In case $u$ and $v$ are not adjacent, we get
$(i-i')^2 + (j-j')^2 \ge 2$ and thus $\norm{p_u-p_v}^2 \ge 8/n^6$.
Together with the bound $2 \le 2y_uy_v < 3$, we obtain
\begin{equation*}
  \frac{\norm{p_u-p_v}^2}{2y_uy_v} \le \frac{4}{2n^6} = \frac{2}{n^6}
  \text{ if } \{u, v\} \in E \quad\text{ and }\quad
  \frac{\norm{p_u-p_v}^2}{2y_uy_v} > \frac{8}{3n^6} \text{ if } \{u, v\} \notin E.
\end{equation*}

As $\arcosh$ is strictly increasing (and $2 < 8/3$), this already
shows that there exists a disk radius $r$ such that we get a uniform
disk representation of $\Gamma_k$.  It remains to show that
specifically $r = 1/n^3$ works.

For this, we use that the series expansion of $\arcosh(1 + x)$ gives
\begin{equation*}
  \sqrt{2x} - \frac{x^{3/2}}{6\sqrt{2}} \le \arcosh(1 + x) \le \sqrt{2x}.
\end{equation*}
If $u$ and $v$ are adjacent, this yields
$d_{\Hyp^2}(p_u,p_v) \le \arcosh(1 + 2 / n^6) \le 2/n^3$.  Thus, the
two disks with radius $1/n^3$ intersect.  If $u$ and $v$ are not
adjacent, we obtain
\begin{equation*}
  d_{\Hyp^2}(p_u,p_v)
  > \arcosh(1+\frac{8}{3n^6})
  \ge \sqrt{2\frac{8}{3n^6}} - \frac{\left(\frac{8}{3n^6}\right)^{3/2}}{6\sqrt{2}}
  > \frac{2.3}{n^3} - \frac{0.6}{n^9}.
\end{equation*}
For $n \ge 4$, this is clearly greater than $2/n^3$ and thus the two
disks of radius $1/n^3$ do not intersect.  Thus, the above defined
coordinates yield an intersection representation of $\Gamma_k$ with
disk of uniform radius $r = 1 / n^3$.

(ii) Suppose for contradiction that $S_n$ has a uniform disk representation with radius $r=1/n^3$ for $n \ge 8$.
Let $c$ be the center vertex of $S_n$.
As $S_n$ has at least $7$ leaves, there are at least two leaves $u$ and $v$ such that the angle between the line segments $cu$ and $cv$ is at most $2\pi /7$.
We will show, that the distance $x$ between $u$ and $v$ is at most $2r$; a contradiction.
As $c$ is connected to $u$ and $v$, both have distance at most $2r$ from $c$.
We first make an argument that the distance between $u$ and $v$ is maximized if either both have distance $2r$ from $c$ or if one has distance $2r$ and the other has the same position as $c$.
Assume we fix the position of $v$ and move $u$ along the line through $u$ and $c$.
Then the distance of $u$ to $v$ is unimodal with a unique minimum and moving $u$ in either direction from the minimum only increases the distance.
Thus, the distance between $u$ and $v$ is maximized by moving $u$ such that it has maximum distance from $c$ or is equal to $c$.

If $u$ has the same position as $c$, then $u$ and $v$ would be connected.
Thus, it remains to consider the case where $u$ and $v$ both have distance $2r$ from $c$ and thus form an isosceles triangle with side-length $2r$ and angle $\alpha \leq 2\pi / 7$; see Figure~\ref{fig:thm28_star_logn}.
This triangle consists of two right triangles for which we can apply the sine formula for hyperbolic triangles to get
\begin{align*}
  \sin\frac{α}{2} = \frac{\sinh\frac{x}{2}}{\sinh 2r}
  \quad \Leftrightarrow\quad
  \sinh\frac{x}{2} = \sin\frac{α}{2}\cdot\sinh 2r.
\end{align*}
Note that $u$ and $v$ are adjacent if and only if $x\leq 2r$, from which we can derive the following condition.
\begin{align*}
  x\leq 2r \quad\Leftrightarrow\quad \sinh\frac{x}{2} \leq \sinh r \quad\Leftrightarrow\quad \sin\frac{\alpha}{2}\cdot \sinh 2r \leq \sinh r.
\end{align*}
As $\alpha \leq 2\pi/7$, we get $\sin\frac{\alpha}{2} \le 0.44$.  Moreover, $\sinh r / \sinh 2r$ is a decreasing function for $r > 0$ and it holds that $\sinh r / \sinh 2r > 0.49$  for $r \le 1/8^3$.  Thus, the above inequality $\sin\frac{\alpha}{2}\cdot \sinh 2r \leq \sinh r$ holds and hence $x \le 2r$, which means that $u$ and $v$ are connected; a contradiction.

\begin{figure*}[tbp]
  \centering
  \begin{minipage}[b]{0.45\textwidth}
    \centering
    \includegraphics[]{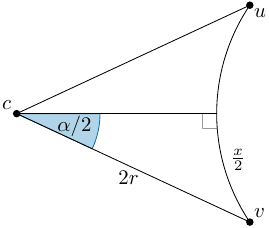}
    \subcaption{}
    \label{fig:thm28_star_logn:a}
  \end{minipage}
  ~
  \begin{minipage}[b]{0.45\textwidth}
    \centering
    \includegraphics[]{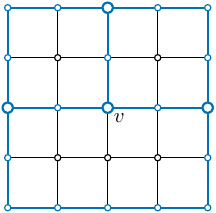}
    \subcaption{}
    \label{fig:thm28_star_logn:b}
  \end{minipage}
  \caption{Visualizations for the proof of
    Theorem~\ref{thm:gridandstar}: (a) drawing of center $c$ and
    leaves $v$, $w$ together with the relevant right triangle in an
    enlarged detail view of the Poincaré disk model, relevant in part
    (ii); (b) subgraph $H$ of the $5 \times 5$ grid with highlighted
    degree $3$ vertices, relevant in part (iii).}
  \label{fig:thm28_star_logn}
\end{figure*}

(iii) We define a subgraph $H$ as follows, see also
Figure~\ref{fig:thm28_star_logn:b}.  Consider the standard Euclidean
drawing of the $5\times 5$ grid and let $C$ denote the cycle of length
16 given by the set of edges on the outer face, and let $v$ be the
vertex in the center of the grid.  Consider $C$ and the $2$-hop paths
connecting $v$ to the middle vertices of the top, left, and right
sides of the square of $C$.  These paths together with $C$ form a
subdivision of a $4$-clique; let $H$ be this subgraph of $\Gamma_5$.

Consider now a realization of $\Gamma_k$ for $k \ge 5$ with hyperbolic
disks of radius $\log n$. We find $H$ as an induced subgraph of
$\Gamma_k$ (as $\Gamma_5$ is a subgraph of $\Gamma_k$) and the
corresponding disks in the realization of $\Gamma_k$ realize
$H$. Consider the drawing of $H$ obtained by connecting the centers of
adjacent disks in $H$ with straight line segments. This drawing is
planar for the following reason. Assume two edges $\{a, b\}$ and
$\{c, d\}$ intersect in a point $x$ and assume without loss of
generality that $a$ is the vertex among $\{a, b, c, d\}$ that is
closest to $x$. Then it follows from the triangular inequality that
$a$ must have an edge to $c$ and $d$, a contradiction to the fact that
$H$ has no triangles.

Thus, we obtain a plane straight-line drawing of a subdivided
$4$-clique and one of the vertices of $H$ (say, $w$) of degree $3$
cannot be on the outer face. Notice that the outer face is bounded by
a cycle $C_w$ that is disjoint from the neighborhood of $w$.  Thus,
$C_w$ forms a polygon and the disk of $w$ must be fully contained in
this polygon. Since $C_w$ has at most $16$ vertices, the area of its
polygon is at most that of $14$ hyperbolic triangles, i.e., at most
$14\pi$, which is less than the area of a disk of radius $\log n$ if
$n \ge 25$; a contradiction.

(iv) We may assume $n\geq 3$ and realize $S_{n}$ with radius $r=\log n$, by placing the center $c$ at the origin and arranging the leaves evenly on a circle with radius $2r$ around $c$.
Clearly, $c$ is adjacent to the leaves.
It remains to show that the leaves are pairwise non-adjacent, i.e., they have distance larger than $2r$.
Consider two neighboring leaves $u$ and $v$ on the circle and the line segments $cu$ and $cv$.
We obtain the isosceles triangle from Figure~\ref{fig:thm28_star_logn} as in the proof of (ii) with side-length $2r$ and angle $\alpha = \frac{2\pi}{n-1}$.
To show that the distance between $u$ and $v$ is greater than $2r$, we have to confirm that $\sin \frac{\alpha}{2}\cdot\sinh 2r> \sinh r$.
Using definition of $\sinh$, $r = \log n$, and $\sin x \ge 2x/\pi$ for $x \in [0, \pi/2]$, we obtain
\begin{align*}
  \sin \frac{\alpha}{2}\cdot \sinh 2r=
  \sin \frac{\pi}{n-1}\cdot \left(\frac{n^2}{2}-\frac{1}{2n^2}\right) > 
  n-\frac{1}{n^3} > \frac{n}{2}-\frac{1}{2n} = \sinh r,
\end{align*}
which concludes the proof.
\end{proof}

\section{Balanced separators in HUDGs}\label{sec:sep}

\subsection{Covering the separator with cliques}
\label{sec:cover-separ-with-cliques}

We start by defining the boxes between a hypercycle and its axis $m$
as illustrated in Figure~\ref{fig:separator-boxes}.  We place points
$b_0, \dots, b_k$ along $m$, such that $b_0 = p$ and they are evenly
spaced at distance~$\tanh(r)$.  Through each point $b_i$ we draw a
line perpendicular to $m$ and let $a_i$ denote its intersection with
the upper hypercycle, see Figure~\ref{fig:separator-boxes}.  We call
the region bounded by the line segments $a_{i - 1}b_{i - 1}$,
$b_{i - 1}b_i$, $b_ia_i$, and the part of the hypercycle between $a_i$
and $a_{i - 1}$ a \emph{box}.

\begin{lemma}
  The diameter of each box is at most $2r$.
\end{lemma}
\begin{proof}
  All boxes are congruent as $b_{i - 1}a_{i - 1}$ and $a_ib_i$ have length $r$, segment $b_{i-1}b_i$ has length $\tanh r$, and there are right angles at $b_{i-1}$ and $b_i$.
  Thus, without loss of generality, we consider only $i = 1$.

  We first argue that the points realizing the diameter must both be vertices.
  For this, fix a point $q$ of the box.
  Let $d$ be the maximum distance of $q$ to a vertex of the box, i.e., the closed disk $D$ of radius $d$ centered at $q$ contains each of the vertices $a_0, a_1, b_0, b_1$ (with one of them lying on its boundary).
  As the sides $a_0b_0$, $b_0b_1$, and $b_1a_1$ are straight line segments and the disk $D$ is convex, these three sides of the box are also contained in $D$ and thus have distance at most $d$ to $q$.
  The side between $a_0$ and $a_1$ is not a straight line but a piece of a hypercycle.
  To see that this side lies also inside $D$, we use that in the Poincaré disk, a hypercycle is a circular arc\footnote{Here we assume without loss of generality that the axis of the hypercycles goes through the center of the Poincaré disk.} as shown in Figure~\ref{fig:separator}.
  Its Euclidean radius in the Poincaré disk is bigger than $1$ (i.e., bigger than that of the Poincaré disk itself).
  Moreover, the disk $D$ is also a disk in the Poincaré disk, but with smaller radius.
  Thus, if $a_0$ and $a_1$ lie in $D$, then the hypercycle segment between $a_0$ and $a_1$ lies in $D$ and thus has distance at most $d$ form $q$.
  It follows that the point of the box with maximum distance to $q$ is one of the vertices, which proves the claim that the diameter is realized by two vertices.

  It remains to bound the distances between vertices.
  The length of both sides $a_0b_0$ and $a_1b_1$ is $r$ and the length of the side $b_0b_1$ is $\tanh r < r$.
  By the triangle inequality, the diagonals have length less than $2r$.
  It thus remains to bound the length of $a_0a_1$.

  To calculate the distance from $a_0$ to $a_1$, we can use that the
  polygon $a_0, b_0, b_1, a_1$ is a \emph{Saccheri quadrilateral} with
  \emph{base} $b_0b_1$ of length $\tanh r$, \emph{legs} $b_0a_0$ and
  $b_1a_1$ of length $r$, and \emph{summit} $a_0a_1$.  The length of
  the summit is given by the following
  formula~\cite[Theorem~10.8]{Euclid_Non_Euclid_Geomet-Green93}:
  \begin{align*}
    \sinh \frac{|a_0a_1|}{2}
    &= \cosh r \sinh \frac{\tanh r}{2} \\
    &< \cosh r \sinh \frac{\min\{1,r\}}{2},
      \intertext{where we used the simple bound $\tanh r<\min\{r,1\}$ for $r>0$ and the fact that $\sinh$ is strictly increasing.
      If $r \leq 1$, then we get}
      \sinh \frac{|a_0a_1|}{2}&< \cosh r \sinh(r/2)\\
    &<2\cosh(r/2)\sinh(r/2)\\
    &=\sinh r,
      \intertext{
      as $\cosh r <2 < 2\cosh(r/2)$ when $0< r \leq 1$. If $r>1$, then
      }
      \sinh \frac{|a_0a_1|}{2} &< \cosh r \sinh(1/2)<\sinh r,
  \end{align*}
  since $\sinh r/\cosh r=\tanh r>\tanh 1>\sinh(1/2)$ as $\tanh r$ is
  monotone increasing. Consequently, $|a_0a_1| < 2r$.  Thus, all
  distances between vertices of a box are at most $2r$, which
  concludes the proof.
\end{proof}

From this lemma, it follows that the vertices inside any box induce a
clique.  Next, we aim to give an upper bound on the number of boxes we
need to cover the separator.  For this, let $k$ be the smallest number
such that $a_k$ lies in the wedge formed by $\ell_1$ and $\ell_2$.  In
Figure~\ref{fig:separator-boxes}, $k = 3$ as $a_3$ lies in the gray
wedge. Clearly, we can cover the whole separator (blue in
Figure~\ref{fig:separator}), with $\Oh(k)$ boxes on both sides of the
line and thus with $\Oh(k)$ cliques.  The following lemma gives a
bound on $k$ depending on the opening angle of the wedge.

\begin{lemma}
  Let $k$ be the smallest number such that $a_k$ lies inside the wedge with opening angle $2\varphi$.
  Then the distance between $b_0$ and $b_k$ is in $\Oh\big(\log\frac{1}{\varphi}\big)$.
\end{lemma}
\begin{proof}
  Consider the triangle in Figure~\ref{fig:separator-triangle}, where $p = b_0$ is the apex of the wedge, $d$ is the intersection of the wedge boundary $\ell_1$ with the hypercycle, and $c$ is the point on the axis with distance $r$ from $d$.
  Recall that the distance between $b_{i - 1}$ and $b_i$ is $\tanh r < 1$.
  Thus, the distance between $b_0$ and $b_k$ equals the distance $x$ between $p$ and $c$ up to an additive constant.
  It remains to give an upper bound on $x$.

  Using hyperbolic trigonometry of right triangles, we get
  \begin{equation*}
    \tan \varphi = \frac{\tanh r}{\sinh x} \quad\Rightarrow\quad \sinh x = \frac{\tanh r}{\tan \varphi}.
  \end{equation*}
  Using that $\tanh r < 1$ and $\tan \varphi > \varphi$, we obtain $\sinh x < \frac{1}{\varphi}$.
  This yields the claim.
\end{proof}

As the distance between $b_{i - 1}$ and $b_i$ is $\tanh r$, we can cover the separator with $\Oh(\log \frac{1}{\varphi} / \tanh r) = \Oh(\log\frac{1}{\varphi}\cdot (1 + 1/r))$ many boxes, each of which is covered by one clique.
This directly yields the following corollary.

\begin{corollary}
  \label{cor:line-cliques}
  Let $G$ be a hyperbolic uniform disk graph with radius $r$.
  Let $m$ be a line through a point $p \in \Hyp^2$ such that all lines through $p$ and a vertex of $G$ have an angle of at least $\varphi$ with $m$.
  Then the subgraph of $G$ induced by vertices of distance at most $r$ from $m$ can be covered with $\Oh\big(\log\frac{1}{\varphi} \cdot (1 + \frac{1}{r})\big)$ cliques.
\end{corollary}

\subsection{Balanced separators}
\label{sec:balanced-separators}

Corollary~\ref{cor:line-cliques} already yields a separator that can
be covered with few cliques, if the angle $\varphi$ is not too small.
It remains to provide the point $p$ together with the line $m$ through
$p$ such that the separator is balanced and $\varphi$ is not too
small.  For this, let $V$ be a set of $n$ points (the vertices of
$G$).  We call a point $p \in \Hyp^2$ a \emph{centerpoint} of $V$ if
every line through $p$ divides $V$ into two subsets that both have
size at least $\frac{n}{3}$.  The following lemma provides such a
centerpoint (which follows directly from the existence of a Euclidean
centerpoint~\cite{centerpoint}). It was also observed
in~\cite{Kisfaludi-Bak21}, but we provide a proof for completeness.

\begin{lemma}[See also~{\cite[proof of Lemma 4]{Kisfaludi-Bak21}}]
  A centerpoint of a set of $n$ points in the hyperbolic plane exists
  and can be found in $\Oh(n)$ time.
\end{lemma}
\begin{proof}
  We do this by first converting the points to the Beltrami-Klein
  model of the hyperbolic plane. Here, any hyperbolic line is also a
  Euclidean line (see also~\cref{fig:prelim}). This means that a
  Euclidean centerpoint is also a hyperbolic centerpoint. Thus, we can
  simply find the Euclidean centerpoint (which takes $\Oh(n)$
  time~\cite{centerpoint}) and then convert it back to our original
  model of the hyperbolic plane.
\end{proof}

Now, consider the lines $\ell_1, \dots, \ell_n$, where $\ell_i$ goes
through $p$ and the vertex $v_i \in V$.  Let without loss of
generality $\ell_1$ and $\ell_2$ be the two consecutive lines with
maximum angle between them.  As there are $2n$ angles between
consecutive lines covering the full $2\pi$ angle, the angle between
$\ell_1$ and $\ell_2$ is at least $\pi / n$.
Set $m$ to be the angular bisector of $\ell_1$ and $\ell_2$.  It
follows that $\frac{1}{\varphi} \in \Oh(n)$ and thus
Corollary~\ref{cor:line-cliques} yields a separator that can be
covered by $\Oh(\log n \cdot (1 + \frac{1}{r}))$ cliques.  Moreover,
as $m$ goes through the centerpoint $p$, this separator is balanced.
Clearly, the line $m$ and thus the separator can be computed in
$\Oh(n \log n)$ time.  This concludes the proof of
Theorem~\ref{thm:separator}.

\thmSeparator*

\section{Outerplanarity of hyperbolic Delaunay complexes}\label{sec:outerplanar}

Before we start with the proof, we make an observation.  If the
sites $S$ are in general position, i.e., no four sites lie on the same
circle, then the Voronoi diagram is a $3$-regular graph and the
Delaunay complex is internally triangulated, i.e., each inner face is
a triangle.  Otherwise, we can slightly perturb the sites $S$ to
get a set $S'$ in general position, such that $\mathcal D(S')$ is obtained from
$\mathcal D(S)$ by triangulating each inner face.  Moreover
$\mathcal V(S')$ is obtained from $\mathcal V(S)$ by splitting each
vertex of degree more than $3$ into a binary tree.  Triangulating
$\mathcal D(S)$ only adds edges to it and thereby only increases its
outerplanarity.  Thus, any upper bound on the outerplanarity of
$\mathcal D(S')$ also holds for $\mathcal D(S)$.  For this section, we
assume without loss of generality that $S$ is in general position,
i.e., $\mathcal D(S)$ is internally triangulated.

\subsection{Large-radius case}
\label{sec:large-radius-case}

We start with the argument for the case where the radius is sufficiently large.
The core observation that any inner vertex of $\mathcal D(S)$ requires many other sites around it is summarized in the following lemma.

\begin{lemma}
  \label{lem:number-of-voronoi-neighbors}
  Let $S$ be a set of sites in $\Hyp^2$ with pairwise distance at least $2r$.
  Any inner vertex of the Delaunay complex $\mathcal D(S)$ has degree at least $e^r$.
\end{lemma}
\begin{proof}
  Let $s\in S$ be a site that is an inner vertex of $\mathcal D(S)$, i.e., its Voronoi region is bounded.
  We show that the angle between consecutive neighbors of $s$ must not to be too large as the Voronoi region of $s$ would be unbounded otherwise.
  Consider two consecutive neighbors $a, b\in S$ of $s$ as shown in Figure~\ref{fig:angle-of-parallelism}.
  Let $p_a$ and $p_b$ be the perpendicular bisectors of $sa$ and $sb$, respectively, and let $\ell$ be the angular bisector of $sa$ and $sb$.
  Since $s$ is an inner vertex, $p_a$ and $p_b$ have to intersect and thus $p_a$ or $p_b$ has to intersect $\ell$.

  \begin{figure}[tbp]
    \centering
    \includegraphics{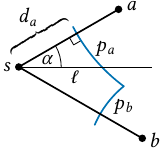}
    \caption{ Illustration of
      Lemma~\ref{lem:number-of-voronoi-neighbors}.  An inner vertex of
      the Delaunay complex $s$ with two consecutive neighbors $a$ and
      $b$.  From the fact that the Voronoi cell of $s$ is bounded, we
      can derive an upper bound for the angle $2\alpha$ between $sa$
      and $sb$.}
    \label{fig:angle-of-parallelism}
  \end{figure}

  Without loss of generality assume that $p_a$ intersects $\ell$, as in Figure~\ref{fig:angle-of-parallelism}.
  Because all sites in $S$ have pairwise distance at least $2r$, the distance $d_a$ between $s$ and the intersection of $sa$ with its perpendicular bisector $p_a$ is at least $r$.
  The angle $\alpha$ between $sa$ and $\ell$ has to be smaller than the angle of parallelism $\Pi(d_a)$, as $\ell$ and $p_a$ intersect.  Thus, we get
  \begin{align*}
    \alpha < \Pi(d_a)\leq \frac{\pi}{2}\sin \Pi(d_a)=\frac{\pi}{2}\frac{1}{\cosh d_a}\leq \pi e^{-d_a}\leq \pi e^{-r}.
  \end{align*}
  Therefore, the angle between $sa$ and $sb$ is at most $2\alpha \leq 2\pi e^{-r}$.
  It follows that $s$ has at least $\frac{2\pi}{2\pi e^{-r}}=e^r$ neighbors.
\end{proof}

Observe that if $r > \log 6$, then every inner vertex of $\mathcal D(S)$ has degree more than $6$.
As the average degree in planar graphs is at most $6$, a constant fraction of vertices have to be outer vertices to make up for the above-average degree of inner vertices.
This observation yields the following lemma.

\begin{lemma}
  \label{lem:highdegree-outerplanar}
  Let $G$ be a plane graph in which every inner vertex has degree at least $d > 6$.
  Then $G$ is $k$-outerplanar for $k < 1 + \log_{d/6} n$.
\end{lemma}
\begin{proof}
  In a planar graph with $n$ vertices, Euler's formula implies that the number of edges is less than $3n$.
  Thus the sum of all vertex degrees is less than $6n$.
  Let $n_{\text{inner}}$ be the number of inner vertices.
  As the sum of inner degrees, which is at least $d \cdot n_{\text{inner}}$, is at most the sum of all degrees, we get $d \cdot n_{\text{inner}} \le 6n$.
  Consequently, $n_{\text{inner}} < n \cdot 6 / d$.

  It follows that removing all outer vertices yields a plane graph with less than $n \cdot 6 / d$ vertices in which every inner vertex again has degree at least $d$.
  Thus, repeatedly removing all outer vertices $k$ times yields a graph with less than $n \cdot (6 / d)^k$ vertices.
  Hence $G$ is reduced to at most one vertex in less than $\log_{d/6} n$ rounds, which concludes the proof.
\end{proof}

For $r \ge 1.8 > \log 6$, these two lemmas already yield the claim of Theorem~\ref{thm:treewidth-of-voronoi-diagram}; also see the formal proof in Section~\ref{sec:combining-small-and-large-radius}, where we combine the large- and small-radius case.

\subsection{Small-radius case}
\label{sec:small-radius-case}

Before formalizing the proof idea, we show two simple results
about hyperbolic triangle areas, both based on the fact that a
right-angled triangle with short sides of length $a$ and $b$ has area
$2\arctan \left( \tanh\frac{a}{2} \cdot \tanh\frac{b}{2} \right)$.  We
start with an upper bound on the area of a triangle where one side has
a specific length.  This will serve as an upper bound for the
triangles in
Figure~\ref{fig:outerplanarity_small_radius_area_polygon}.

\begin{figure}[tbp]
  \centering
  \begin{subfigure}[b]{0.32\textwidth}
    \centering
    \includegraphics[page=1]{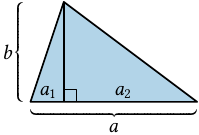}
    \subcaption{\label{fig:outerplanarity_triangles_perpendicular_inside}}
  \end{subfigure}
  \begin{subfigure}[b]{0.32\textwidth}
    \centering
    \includegraphics[page=2]{figures/outerplanarity_triangles}
    \subcaption{\label{fig:outerplanarity_triangles_perpendicular_outside}}
  \end{subfigure}
  \caption{Illustration of Lemma~\ref{lem:trianglearea1} giving an
    upper bound on the area of a triangle (blue) with side length
    $a$.}
  \label{fig:outerplanarity_triangles}
\end{figure}

\begin{lemma}
  \label{lem:trianglearea1}
  Any triangle with one side of length $a$ has area at most
  $\min\{a, \pi\}$.
\end{lemma}
\begin{proof}
  We first need that $\arctan(\tanh x) \leq \min\{x, \pi/4\}$ for
  $x \geq 0$.  Here, $\tanh x < 1$ implies
  $\arctan(\tanh x) < \arctan1 = \pi/4$ as $\arctan$ is strictly
  increasing.  For the other part, using that $\tanh(x) \le x$ and
  $\arctan(x) \le x$ for $x \ge 0$ yields
  $\arctan(\tanh x) \le \tanh x \le x$.

  Consider the line perpendicular to the side of length $a$ through
  its opposite vertex.  It either splits the triangle in two right
  triangles as shown in
  Figure~\ref{fig:outerplanarity_triangles_perpendicular_inside} or it
  extends the triangle into a larger one as shown in
  Figure~\ref{fig:outerplanarity_triangles_perpendicular_outside}.  In
  both cases, let $b$ be the length of the new perpendicular line.
  Moreover, in the first case, the side of length $a$ is split into
  two sides of length $a_1$ and $a_2$.  In the second case, the side
  of length $a$ is extended by $c$ yielding a right triangle with side
  length $c + a$.

  For the first case, the area of the triangle (blue in
  Figure~\ref{fig:outerplanarity_triangles_perpendicular_inside}) is
  the sum of the two triangles, i.e.,
  \begin{align*}
    2\arctan \left( \tanh\frac{a_1}{2} \tanh\frac{b}{2} \right) &+ 
    2\arctan \left( \tanh\frac{a_2}{2} \tanh\frac{b}{2} \right)\\
    &\leq 2\arctan \left( \tanh\frac{a_1}{2} \right)
      + 2\arctan \left( \tanh\frac{a_2}{2} \right) \\
    &\leq \min\{a_1, \pi/2\} + \min\{a_2, \pi/2\} \\
    &\leq \min\{a, \pi\}.
  \end{align*}

  For the second case, we are interested in the area of the blue
  triangle in
  Figure~\ref{fig:outerplanarity_triangles_perpendicular_outside}.
  Using that $\arctan(x)$ and $\tanh(x)$ are both subadditive\footnote{A
    function $f$ is subadditive if $f(x + y) \le f(x) + f(y)$} for $x\geq 0$ and
  $\arctan$ is increasing, we can estimate the area of the right
  triangle with side lengths $b$ and $c + a$
  ($\text{red} + \text{blue}$ in
  Figure~\ref{fig:outerplanarity_triangles_perpendicular_outside}) as
  \begin{align*}
    2\arctan\left(\tanh \frac{a + c}{2} \tanh \frac{b}{2} \right)
    &\le 2\arctan\left(\left(\tanh \frac{a}{2} + \tanh \frac{c}{2}\right) \tanh \frac{b}{2} \right)\\
    &\le 2\arctan\left(\tanh \frac{a}{2} \tanh \frac{b}{2} \right) + 2\arctan\left(\tanh \frac{c}{2} \tanh \frac{b}{2} \right).
  \end{align*}
  Now we get the area of the triangle with side length $a$ (blue), as
  the difference between two right triangles
  ($\text{red} + \text{blue} - \text{red}$), which yields
  \begin{align*}
    2\arctan\left(\tanh \frac{a + c}{2} \tanh \frac{b}{2} \right) -
    2\arctan\left(\tanh \frac{c}{2}  \tanh \frac{b}{2} \right)
    &\le 2\arctan\left(\tanh \frac{a}{2}  \tanh \frac{b}{2} \right)\\
    &\le 2\arctan\left(\tanh \frac{a}{2} \right)\\
    &\le \min\{a, \pi\}.\qedhere
  \end{align*}
\end{proof}

To lower bound the area of the triangles in
Figure~\ref{fig:outerplanarity_small_radius_area_layer}, we use the
following lemma.  We note that the requirement $a \le 1.8$ in this
lemma is the reason why this section only deals with the case
$r \le 1.8$.

\begin{lemma}
  \label{lem:trianglearea2}
  A right-angled triangle with short sides of length $a \leq 1.8$ and $b$
  has area at least $a/5 \cdot \min\{b, \pi\}$.
\end{lemma}
\begin{proof}
	The area is given by the function
	$f_b(a) := 2\arctan \left( \tanh\frac{a}{2} \cdot \tanh\frac{b}{2} \right)$,
	which is concave, as
	$f_b''(a) = -\frac{2 \sinh a \sinh(2b)}{(2 + 2\cosh a \cosh b)^2} \leq 0$
	for any $a,b \in \mathbb R$. Because additionally $f_b(0) = 0$,
	we can say that $f_b(a) \geq a \cdot f_b(1.8) / 1.8$ for $a \in [0,1.8]$.
	What remains is to bound $g(b) := f_b(1.8)$. This function $g(b)$ is
	concave for analogous reasons as before, and again $g(0) = 0$, so we can
	say that $g(b) \geq b \cdot g(\pi) / \pi \geq 0.37 b$ for $b \in [0,\pi]$.
	Additionally, $g(b) \geq 0.37 \pi$ for $b \geq \pi$.
	Thus, $f_b(a) \geq a / 5 \cdot \min\{b, \pi\}$ for $a \in [0,1.8]$.
\end{proof}

With this, we are ready to show that the area of the layer polygons
grows exponentially.

\begin{lemma}
  \label{lem:outer-planarity-small-radius}
  Let $r \le 1.8$ and let $S$ be a set of sites in $\Hyp^2$ with
  pairwise distance at least $2r$.  Let $s \in S$ be any site and let
  $P_\ell$ be the $\ell$-th layer polygon of $s$ in the Delaunay
  complex $\mathcal D(S)$.  Then
  $\area(P_\ell) \ge \frac{10}{10 - r} \cdot \area(P_{\ell - 1})$.
\end{lemma}
\begin{proof}
  We give an upper bound on $\area(P_\ell)$ as well as a lower bound
  on $\area(P_\ell) - \area(P_{\ell - 1})$ and then relate these two.
  For the upper bound on $\area(P_\ell)$, we cover $P_\ell$ with
  triangles as shown in
  Figure~\ref{fig:outerplanarity_small_radius_area_polygon} and sum
  the area of these triangles.  For each $v_i \in V_\ell$, we get two
  triangles with sides $v_\ell^{i - 0.5} v_\ell^{i}$ and
  $v_\ell^i v_\ell^{i + 0.5}$, respectively.  Using
  Lemma~\ref{lem:trianglearea1} to bound their area, we obtain
  \begin{equation*}
    \area(P_\ell) \leq \sum_{i=1}^{|V_\ell|} \min\{|v_\ell^{i - 0.5}
    v_\ell^{i}|, \pi\} + \min\{|v_\ell^{i} v_\ell^{i + 0.5}|, \pi\}.
  \end{equation*}

  For the lower bound, let $v_i \in V_\ell$ and consider the triangle
  with a right angle at $v_\ell^i$ where the two sides forming that
  right angle are $v_\ell^iv_\ell^{i + 0.5}$ and the perpendicular
  line segment of length $r$ that lies in the interior of $P_\ell$;
  see Figure~\ref{fig:outerplanarity_small_radius_area_layer}.
  Analogously, we define a triangle with side
  $v_\ell^{i + 0.5}v_\ell^{i}$.  Using Lemma~\ref{lem:trianglearea2}
  the larger of the two triangles has area at least
  $r/5 \cdot \min \big\{ \max\{|v_\ell^{i-0.5} v_\ell^i|, |v_\ell^i
  v_\ell^{i+0.5}|\},\ \pi \big\}$, which is at least
  $r/10 \cdot \big( \min\{|v_\ell^{i-0.5} v_\ell^i|, \pi\} +
  \min\{|v_\ell^i v_\ell^{i+0.5}|, \pi\} \big)$.  Note that these two
  triangles lie inside $P_\ell$ but outside $P_{\ell - 1}$ and they do
  not intersect other triangles that are obtained in the same way for
  other vertices in $V_\ell$.  Thus, we obtain
  \begin{equation*}
    \area(P_\ell) - \area(P_{\ell - 1}) \ge \frac{r}{10}
    \sum_{i=1}^{|V_\ell|} \min\{|v_\ell^{i - 0.5} v_\ell^{i}|, \pi\} +
    \min\{|v_\ell^{i} v_\ell^{i + 0.5}|, \pi\}.
  \end{equation*}

  Observe that this lower bound is the same as the upper bound except
  for the factor of $r / 10$.  Thus, we obtain
  $\area(P_\ell) - \area(P_{\ell - 1}) \ge \frac{r}{10}
  \area(P_\ell)$.  Rearranging yields the claimed bound.
\end{proof}

\subsection{Combining the two}
\label{sec:combining-small-and-large-radius}

Combining the results for small and large radius, we obtain the
desired theorem.

\thmOuterplanarity*
\begin{proof}
  First, assume $r \ge 1.8$.
  As $1.8 > \log 6$, it follows from Lemma~\ref{lem:number-of-voronoi-neighbors} that any inner vertex of $\mathcal D(S)$ has degree at least $e^r > 6$.
  We can thus apply Lemma~\ref{lem:highdegree-outerplanar} to obtain that $\mathcal D(S)$ is $k$-outerplanar for $k < 1 + \log_{e^r / 6} n$, which matches the claimed bound.

  For the case where $r \le 1.8$, first note that the claimed bound
  trivially holds for $r \le 1/n$ as any planar graph is clearly
  $\Oh(n \log n)$-outerplanar.  For $1/n < r \le 1.8$, let $s \in S$
  be any site and consider the layer polygons $P_1, \dots, P_L$ where
  $L + 1$ is the distance of $s$ in $\mathcal D(S)$ to the closest
  outer vertex.  Then by
  Lemma~\ref{lem:outer-planarity-small-radius} the area of the
  polygons grows by a factor of at least $\frac{10}{10 - r}$ in every
  step and thus we get
  \begin{equation*}
    \left( \frac{10}{10 - r} \right)^L \cdot \area(P_1)
    \le \area(P_L) \quad\Rightarrow\quad
    L \le \frac{\log \frac{\area(P_L)}{\area(P_1)}}{\log\frac{10}{10 - r}}.
  \end{equation*}

  First note that $\log\frac{10}{10 - r} \ge \frac{r}{10}$ for
  $r < 10$ and thus it remains to show that
  $\frac{\area(P_L)}{\area(P_1)}$ is polynomial in $n$.  For this
  observe that $\area(P_L) \le 2 \pi n$ as it is a polygon with
  less than $2n$ vertices.  Moreover, $P_1$ at least includes the
  Voronoi cell of $s$, which includes a disk of radius $r$ as sites
  have pairwise distance at least $2r$.  Thus,
  $\area(P_1) \ge 4\pi \sinh^2(r / 2) \in \Omega(1/n^2)$ as
  $r \ge 1/n$ and $\sinh x$ behaves like $x$ for $x$ close to $0$.  It
  follows that $L \in \Oh\big(\frac{\log n}{r}\big)$, which concludes
  the proof.
\end{proof}

\section{Exact algorithm for independent set}
\label{sec:exact_algo_appendix}

In the following, we give an exact algorithm for computing an
independent set of a given size in a hyperbolic uniform disk graph $G$
of radius $r$ that is given via its geometric realization.  Let $S$ be
an independent set of $G$.  As there are no edges between vertices in
$S$, they have pairwise distance at least $2r$.  Thus,
Theorem~\ref{thm:treewidth-of-voronoi-diagram} gives us an upper bound
on the outerplanarity of the Delaunay complex $\cD(S)$, which implies
that $\cD(S)$ has small treewidth and thereby small balanced
separators.  In a nutshell, our algorithm first lists all relevant
separators of the Delaunay complexes of all possible independent sets.
We then use a dynamic program to combine these candidate separators
into a hierarchy of separators, maximizing the number of vertices of
the corresponding Delaunay complex and thereby the size of the
independent set $S$.

A crucial tool for dealing with these separators are sphere cut
decompositions.  They are closely related to tree decompositions but
represent the separators as closed curves called \emph{nooses}.  We
introduce sphere cut decompositions in
Section~\ref{sec:sphere-cut-decomp}.  Afterwards, in
Section~\ref{sec:exact_algo:delaunay} we
define normalized geometric realizations of nooses for sphere cut
decompositions of Delaunay complexes.  In
Section~\ref{sec:exact-algo:cand-noose}, we prove that there are not
too many candidate nooses and that selecting a subset of candidate
nooses that can be combined into a hierarchy of separators actually
yields an independent set.  Finally, in
Section~\ref{sec:exact-algo:solv-indep-set}, we describe the dynamic
program for computing the independent set, which follows almost
immediately from the previous sections.

\subsection{Sphere cut decompositions}
\label{sec:sphere-cut-decomp}

A \emph{branch decomposition}~\cite{Robertson_Seymour_1991} of a graph
$G$ is an unrooted binary tree $T$, i.e., each non-leaf node has
degree exactly 3, and there is a bijection between the leaves of $T$
and the edges of $G$.  Observe that removing an edge $e_T \in E(T)$
separates $T$ into two subtrees and thereby separates its leaves and
thus the edges of $G$ in two subsets.  Let $E_1\subset E(G)$ and
$E_2 \subset E(G)$ be these two edge sets and let $G[E_1]$ and
$G[E_2]$ be their induced subgraphs.  The vertices shared by $G[E_1]$
and $G[E_2]$, i.e., vertices that appear in edges of $E_1$ and of
$E_2$, are called the \emph{middle set} of $e_T$.  The \emph{width} of
the branch decomposition $T$ is the size of the largest middle set
over all edges of $T$.  The \emph{branch-width} of $G$ is the minimum
width of all branch decompositions of $G$.  Robertson and
Seymour~\cite{Robertson_Seymour_1991} showed that the branch-width of
a graph is within a constant factor of its treewidth.

\begin{figure}[tbp]
  \centering
  \includegraphics[]{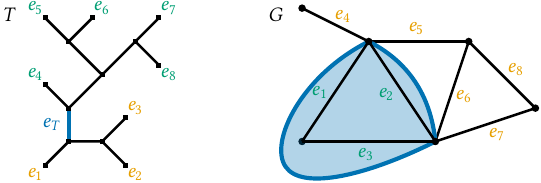}
  \caption{A branch decomposition $T$ of a graph $G$ with a
    highlighted edge $e_{T}$ of $T$ and the corresponding
    noose in $G$, that separates the two parts of $E(G)$.}
  \label{fig:sphere_cut_definition_help}
\end{figure}

If $G$ is a plane graph, results by Seymour and
Thomas~\cite{DBLP:journals/combinatorica/SeymourT94} imply that we can
wish for some additional structure without increasing the width of the
branch
decomposition~\cite{Dorn_Penninkx_Bodlaender_Fomin_2010,MarxP22}.
Specifically, a \emph{sphere cut decomposition} of a plane graph $G$
is a branch decomposition $T$ such that every tree edge $e_T \in E(T)$
is associated with a so-called \emph{noose}.  The noose $O$ of $e_T$
is a simple closed curve that intersects the vertices of $G$ exactly
in the middle set of $e_T$ and separates $G[E_1]$ from $G[E_2]$; see
also Figure~\ref{fig:sphere_cut_definition_help}.  Moreover, removing the
vertices of the middle set from $O$ separates $O$ into \emph{noose
  segments} such that each segment lies entirely in the interior of a
face of $G$ and every face contains at most one noose segment, except
if $E_1$ or $E_2$ contains just a bridge of $G$.  In the latter case,
the noose consists of two noose segments through the face incident to
the bridge.  It should be mentioned that this differs from the
original definition by Dorn, Penninkx, Bodlaender and
Fomin~\cite{Dorn_Penninkx_Bodlaender_Fomin_2010}.  Here, a noose may
intersect any face at most once and thus the graph cannot have
bridges, as discussed by Marx and Pilipczuk~\cite{MarxP22}.  With our
relaxed definition of nooses however, we get the following
theorem.

\begin{theorem}[\cite{Dorn_Penninkx_Bodlaender_Fomin_2010,MarxP22}]
  Let $G$ be a plane graph with branch-width $b$.  Then $G$ has a
  sphere cut decomposition, where every noose intersects at most $b$
  vertices of $G$.
\end{theorem}

In the following, we assume sphere cut decompositions to be rooted at
a leaf node.  Let $e_T \in E(T)$ be a tree edge between $u_T \in V(T)$
and $v_T \in V(T)$ such that $u_T$ is the parent of $v_T$.  Moreover,
let $G[E_1]$ and $G[E_2]$ be the two subgraphs defined by $e_T$ and
assume that $E_1$ corresponds to the leaves that are descendants of
the child $v_T$ in $T$.  Then, we call the side of the noose $O$ of
$e_T$ that contains $G[E_1]$ its \emph{interior} and the side
containing $G[E_2]$ its \emph{exterior}.  We note that, if the root of
$T$ is appropriately chosen, this corresponds to the typical
definition of interior and exterior of closed curves in the plane.  In
the following, when specifying a noose (or a closed curve that could
be a noose), its interior and exterior is implicitly given by assuming
a clockwise orientation, i.e., we assume that the interior of a noose
lies to its right when traversing it in the given orientation.

We note that rooting the sphere cut decomposition also defines a
child--parent relation on the nooses.  Consider a non-leaf node
$v_T \in V(T)$.  It has an edge $e_P$ to the parent, and two edges
$e_L$ and $e_R$ to the left and right child corresponding to the
nooses $O_P$, $O_L$, and $O_R$, respectively.  We say that $O_P$ is
the \emph{parent noose} of the \emph{child nooses} $O_L$ and $O_R$.
Note that the subgraph in the interior of $O_P$ is the union of the
subgraphs in the interior of $O_L$ and $O_R$.

\subsection{Geometric nooses of Delaunay complexes}
\label{sec:exact_algo:delaunay}

In this section, we are interested in the sphere cut decomposition of
the Delaunay complex $\mathcal D(S)$ of a set of sites $S$.  When
working with a sphere cut decomposition, the exact geometry of the
nooses is often not relevant, i.e., it is sufficient to know that
there exists a closed curve with the desired topological properties.
In our case, however, the geometry is actually important.  Our goal
here is the following.  Given a noose $O$ of a sphere cut
decomposition of $\mathcal D(S)$, we define a normalized geometric
realization of $O$, i.e., a fixed closed curve satisfying the noose
requirements.  Afterwards, we will observe that this normalized
realization has some nice properties.

Our normalized nooses will be generalized polygons (recall the
definition in Section~\ref{sec:prelim}).  We note that generalized
polygons are technically not closed curves in the hyperbolic plane.
However, they are closed curves in the Poincaré disk when perceived as
a disk of the Euclidean plane.  Thus, generalized polygons are suited
to represent nooses.

\begin{figure}[tbp]
  \centering
  \begin{subfigure}[b]{0.26\textwidth}
    \centering
    \includegraphics[page=4]{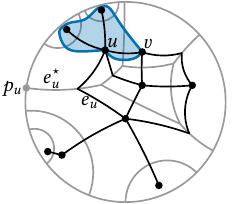}
    \subcaption{}
    \label{fig:nooses:b}
  \end{subfigure}
  \hfill 
  \begin{subfigure}[b]{0.23\textwidth}
    \centering
    \includegraphics[page=2]{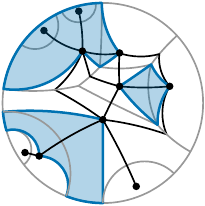}
    \subcaption{}
    \label{fig:nooses:c}
  \end{subfigure}
  \hfill 
  \begin{subfigure}[b]{0.23\textwidth}
    \centering
    \includegraphics[page=3]{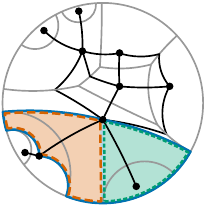}
    \subcaption{}
    \label{fig:nooses:d}
  \end{subfigure}
  \caption{
    Part (a)
    visualizes the choice of the ideal point $p_{u}$ for a Voronoi
    ray, as needed for the normalization of the green noose.  Part (b)
    shows the normalized version of the noose from part (a) alongside
    two other normalized nooses, each separating a single edge of
    $\mathcal{D}(S)$.  Observe that the depicted normalized nooses
    include zero, one, or two ideal arcs, respectively.  Part (c)
    shows a parent noose (blue) and its child nooses, with interiors
    colored green and orange.}
  \label{fig:nooses}
\end{figure}

To define the normalized nooses, let $O$ be a noose and consider an
individual noose segment of $O$ that goes from $u$ to $v$ (which are
vertices of $\cD(S)$) through the face $f$ of $\cD(S)$.  If $f$ is an
inner face, then $f$ corresponds to a vertex of the Voronoi diagram
$\mathcal V(S)$.  Let $p_f$ be the position of this vertex.  Then, the
\emph{normalized noose segment} from $u$ to $v$ through $f$, consists
of the two straight line segments from $u$ to $p_f$ and from $p_f$ to
$v$.

If $f = f_\infty$ is the outer face, $u$ and $v$ can have multiple
incidences to $f_\infty$ (in fact $u$ and $v$ could be the same
vertex).  To make the situation more precise, assume that we traverse
the boundary of $f_\infty$ such that $f_\infty$ lies to the left and
let $e_u$ be the edge incident to $f_\infty$ that precedes the
incidence of $u$ where the noose enters $f_\infty$; see
Figure~\ref{fig:nooses:b}.  Recall that the dual Voronoi edge
$e_u^\star$ is unbounded and ends in some ideal point $p_u$. Let $e_v$
and $p_v$ be defined analogously for $v$.  Then, for the normalized
noose segment from $u$ to $v$ in $f_\infty$, we use the ray between
$u$ and $p_u$, the ideal arc from $p_u$ to $p_v$, and the ray between
$p_v$ and $v$.

Observe that combining the geometric noose segments as defined above
yields a generalized polygon for each noose; see
Figure~\ref{fig:nooses:c}.  Also, note that this definition also works
in the special case where the noose goes only through one vertex with
two incidences to the outer face, in which case the noose consists of
just two rays and an ideal arc.

It is easy to observe that the above construction in fact yields
geometric realizations of the nooses, i.e., the resulting curves have
the desired combinatorial properties in regards to which elements of
$\cD(S)$ lie inside, outside, or on a noose.  We call a sphere
cut decomposition of $\cD(S)$, with such normalized nooses a
\emph{normalized sphere cut decomposition}.  In the following, we
summarize properties of these nooses that we need in later arguments.
The first lemma follows immediately from the construction above.

\begin{lemma}
  \label{lem:normalized_noose_summary}
  Let $S$ be a set of sites with Delaunay complex $\cD(S)$ and Voronoi
  diagram $\mathcal V(S)$.  Let $O$ be a noose of a normalized sphere
  cut decomposition of $\cD(S)$.  Then $O$ is a generalized polygon
  and each vertex of $O$ is either a site, a vertex of
  $\mathcal V(S)$, or an ideal vertex of $\mathcal V(S)$.  Between any
  two subsequent sites visited by $O$, there is either exactly one
  vertex of $\mathcal{V}(S)$ or two ideal vertices of $\mathcal{V}(S)$
  in $V(O)$.  In the latter case, the ideal vertices are connected via
  an ideal arc.
\end{lemma}

The previous lemma summarizes the core properties of individual
nooses. Next, we describe additional properties of how nooses interact
in their parent-child relationships. See also Figure~\ref{fig:nooses:d}.

\begin{lemma}
  \label{lem:normalized_noose_combination}
  Let $S$ be a set of sites with Delaunay complex $\cD(S)$.
  Let $O_P$, $O_L$, and $O_R$ be three nooses of a normalized sphere
  cut decomposition of $\cD$ such that $O_P$ is the parent
  noose of $O_L$ and $O_R$.  Then $O_{L}$, $O_{R}$ and $O_{P}$
  intersect in exactly two points $p_{1}$, $p_{2}$ and
  $O_P \setminus \{p_{1}, p_{2}\}$ is the symmetric difference of $O_L$ and
  $O_R$.  Further, the interiors of $O_L$ and $O_R$ are subsets of the
  interior of $O_P$.
\end{lemma}
\begin{proof}
  Recall that a noose is a simple closed curve intersecting $G$
  precisely at its associated vertex separator.  It can be viewed as a
  sequence of segments, each traversing a face by connecting two
  incident vertices.  For the proof of the lemma, we first take a step
  back from geometric perspective and define the \emph{combinatorial
    representation} of a noose as the (multi-)set of its vertex--face
  incidences.  Note that the combinatorial representation of noose $O$
  is already uniquely determined by its interior and exterior edges.

  We show that the combinatorial representation of $O_{P}$ is the
  symmetric difference of the combinatorial representations of $O_{L}$
  and $O_{R}$.  For this, we partition the edges of $G$ into edges
  $E_{L}$ in the interior of $O_{L}$, $E_{R}$ in the interior of
  $O_{R}$ and $E_{\mathrm{ex}}$ in the exterior of $O_{P}$.  Then the
  edges in the interior of $O_{P}$ are $E_{L} \cup E_{R}$.  We
  consider a vertex--face incidence $i = (v,f)$ of one of the child
  nooses, without loss of generality $O_{L}$.  As $O_{L}$ separates
  $E_{L}$ from $E_{R} \cup E_{\mathrm{ex}}$, next to $i$ on the
  boundary of $f$ there is an edge $e_{1}\in E_{L}$ and an edge
  $e_{2} \in E_{R} \cup E_{\mathrm{ex}}$.  If $e_{2} \in E_{R}$, then
  $O_{R}$ also enters $f$ via $v$ and also contains the vertex--face
  incidence $i$, while $O_{P}$ does not contain $i$.  If otherwise
  $e_{2} \in E_{\mathrm{ex}}$, then $O_{P}$ also visits $f$ through
  the same vertex incidence, while $O_{R}$ does not.  All other cases
  are analogous and in particular any vertex--face incidence of the
  parent noose occurs in exactly one of the child nooses.  It follows
  that the vertex--face incidences of $O_{P}$ are the symmetric
  difference of those of $O_{L}$ and $O_{R}$.

  The normalized geometric realization of a noose contains a line
  segment or ray for every vertex--face incidence.  This line segment
  or ray is uniquely determined by the incidence's vertex and face according to
  the construction of normalized nooses.  As $O_{P}$, $O_{L}$, and
  $O_{R}$ are nooses of a normalized sphere cut decomposition, it
  follows that $E(O_{P})$ is the symmetric difference of $E(O_{L})$
  and $E(O_{R})$.\footnote{Recall that $O_{P}$ is a generalized
    polygon and thus $E(O_{P})$ refers to the edges of a
    plane graph.}  Thus, $O_{P}$ is the symmetric difference of
  $O_{L}$ and $O_{R}$, except for exactly two points that are visited
  by all three normalized nooses.  These correspond either to a vertex
  of $\mathcal{D}(S)$ at which all three nooses meet or to an (ideal)
  Voronoi vertex located in a face of $\mathcal{D}(S)$ that is visited
  by all three nooses.  As the nooses are simple closed curves, it
  directly follows that the interiors of $O_L$ and $O_R$ are subsets
  of the interior of $O_P$.
\end{proof}

Finally, our last lemma states that a normalized noose does not get
too close to sites not visited by that noose.

\begin{lemma}
  \label{lem:normalized_nooses_distance}
  Let $S$ be a set of sites with pairwise distance at least $2r$, let
  $\cD(S)$ be its Delaunay complex, and let $O$ be a noose of a
  normalized sphere cut decomposition of $\cD(S)$.  Then $O$ has
  distance at least $r$ from any site not visited by $O$.
\end{lemma}
\begin{proof}
  Let $s_1, s_2 \in S$ be two consecutive sites visited by $O$ and let
  $O'$ be the segment of $O$ between $s_1$ and $s_2$.  We first argue
  that $O'$ does not go through the interior of any Voronoi cell of a
  site other than $s_1$ or $s_2$.  As $O$ is a noose, $O'$ stays
  inside one face $f$ of $\cD(S)$ with $s_1$ and $s_2$ on its
  boundary.  If $f$ is an inner face, then due to
  Lemma~\ref{lem:normalized_noose_summary}, $O'$ is comprised of the
  two line segments from $s_1$ to the Voronoi vertex corresponding to
  $f$ and from there to $s_2$.  As Voronoi cells are convex, $O'$ does
  not enter any other Voronoi cell. Similarly, if $f$ is the outer
  face, then the line segment of $s_i$ for $i \in \{1, 2\}$ to an
  ideal vertex $q_i$ of $\mathcal V(S)$ on the boundary of the Voronoi
  cell of $s_i$ does not leave the cell of $s_i$. Also, the ideal arc
  between $q_1$ and $q_2$ does not enter the interior of any Voronoi
  cell.

  As this holds for any two consecutive sites visited by $O$, it
  follows that $O$ goes only through Voronoi cells of sites it visits.
  Let $s \in S$ be a site not visited by $O$.  As all other sites have
  pairwise distance at least $2r$ to $s$, the open disk of radius
  $r$ around $s$ lies inside the Voronoi cell of $s$.  As the noose
  does not enter the Voronoi cell of $s$, it does not enter this open
  disk and thus has distance at least $r$ from~$s$.
\end{proof}

\subsection{Candidate nooses and noose hierarchies}
\label{sec:exact-algo:cand-noose}

In this section, we come back to the initial problem of computing an
independent set of a hyperbolic uniform disk graph $G$.  From the previous section, we know
that any independent set $S \subseteq V(G)$ of $G$ has a Delaunay
complex $\cD(S)$ with a sphere cut decomposition in which all nooses
satisfy the properties stated in
Lemmas~\ref{lem:normalized_noose_summary}--\ref{lem:normalized_nooses_distance}.
In the following, we show that these properties are also sufficient in
the sense that a collection of nooses satisfying them corresponds to
an independent set of $G$.  Thus, instead of looking for an
independent set itself, we can search for a hierarchy of nooses that
satisfy these properties.

To make this precise, we call a generalized polygon $O$ a
\emph{$k$-candidate noose} for $G$ if there exists an independent set
$S$ of size $|S| \le k$ of $G$ and a minimum width normalized sphere
cut decomposition of its Delaunay complex $\cD(S)$ that has $O$ as a
noose.  The following lemma states that there are not too many
candidate nooses and that we can compute them efficiently.  It in
particular implies that even though there may be exponentially many
independent sets in $G$, the normalized sphere cut decompositions of
their Delaunay complexes contain substantially fewer different nooses.

\begin{lemma}
  \label{lem:candidate_nooses_upper_bound}
  Let $G$ be a hyperbolic uniform disk graph with radius $r$, let
  $k \le n$ and let $\mathcal C$ be the set of all $k$-candidate
  nooses for $G$.  Then
  $|\mathcal C| \in n^{\Oh(1 + \frac{\log k}{r})}$.  Moreover, a set
  of generalized polygons that contains $\mathcal C$ can be computed
  in $n^{ \Oh(1 + \frac{\log k}{r}) }$ time.
\end{lemma}
\begin{proof}
  Let $S$ be a independent set of $G$ with $|S| \le k$.  As the
  vertices of $S$ have pairwise distance at least $r$, the Delaunay
  complex $\cD(S)$ of $S$ is
  $1+\Oh\big(\frac{\log k}{r}\big)$-outerplanar by
  Theorem~\ref{thm:treewidth-of-voronoi-diagram}.  The outerplanarity
  of a graph gives a linear upper bound on its treewidth~\cite[Theorem
  83]{Bodlaender98} and thus also its
  branchwidth~\cite[(5.1)]{Robertson_Seymour_1991}.  Thus,
  $\mathcal{D}(S)$ has branchwidth $w \in \Oh\big(1+\twformulak\big)$,
  i.e., each noose of any minimum width sphere cut decomposition of
  $\cD(S)$ visits at most $w$ sites. As candidate nooses are required
  to be nooses in minimum width branch decompositions, each candidate
  noose visits at most $w$ vertices.
  Each visited vertex is one of $n$ vertices of $G$ and thus there are
  at most $n^w$ different sequences in which a candidate noose can
  visit at most $w$ vertices of
  $G$. 

  For a fixed sequence of visited vertices, there is the additional
  choice of how the candidate noose gets from one vertex to the next.
  Let $u$ and $v$ be two vertices that are consecutive in this
  sequence of vertices.  As candidate nooses have to be nooses of
  normalized sphere cut decompositions, we can use that any such
  normalized noose satisfies Lemma~\ref{lem:normalized_noose_summary},
  i.e., the normalized noose is a generalized polygon and between $u$
  and $v$ there is either one Voronoi vertex of $\mathcal V(S)$ or two
  ideal Voronoi vertices of $\mathcal V(S)$ with an ideal arc between
  them.  Although we do not know $S$, there are not too many choices
  for (ideal) Voronoi vertices.  Each Voronoi vertex of
  $\mathcal V(S)$ is the unique point that has equal distance from
  three vertices in $S$ and is thus determined by choosing three
  vertices of $G$.  This mean that, although there may be many
  independent sets $S$, there are only $\Oh(n^3)$ positions where
  Voronoi vertices of $\mathcal V(S)$ can be.  Similarly, each ideal
  Voronoi vertex is the ideal endpoint of a perpendicular bisector
  between two vertices of $G$.  Thus, there are only $\Oh(n^4)$ ways
  to choose two ideal Voronoi vertices.  As the noose visits at most
  $w$ vertices, there are up to $w$ pairs of consecutive vertices $u$
  and $v$ and for each we can choose one of the $n^{\Oh(1)}$ ways to
  connect them, amounting to $n^{\Oh(w)}$ options in total.

  To summarize, there are $n^w$ sequences of vertices of $G$ that can
  be visited by a candidate noose.  For each of these sequences, there
  are $n^{\Oh(w)}$ ways of how a candidate noose can connect these
  vertices.  This gives the upper bound of $n^{\Oh(w)}$ for the number
  of different nooses.  As $w \in \Oh\big(1+\twformulak\big)$, this
  gives the desired bound on $|\mathcal C|$.

  Observe that the above estimate is constructive, i.e., we can
 efficiently enumerate all possible options and filter out sequences
 of points that do not give a generalized polygon (e.g., as it would
 be self-intersecting).  Note that not all options actually yield
 valid candidate nooses.  However, we clearly get a superset of
 $\mathcal C$ in the desired time.
\end{proof}

Next, we go from considering individual candidate nooses to how
candidate nooses can be combined into a sphere cut decomposition.  To
this end, let $G$ be a hyperbolic uniform disk graph with radius $r$.
We define a \emph{polygon hierarchy} as a rooted full binary tree
of generalized polygons that visit vertices of $G$ (but can also
contain other (ideal) points).  Each polygon in the hierarchy is
either a leaf or has two child polygons.  Let $O_P$ be a parent
polygon with children $O_L$ and $O_R$ and let
$S(O_P, O_L, O_R) \subseteq V(G)$ be the set of vertices of $G$
visited by at least one of these generalized polygons.  We say that
this child--parent relation is \emph{valid} if $O_{L}$, $O_{R}$ and
$O_{P}$ meet in exactly two points $p_{1}, p_{2}$ such that
$O_P \setminus \{p_{1}, p_{2}\}$ is the symmetric difference of $O_L$
and $O_R$ and the interiors\footnote{As before, we assume that the
  region to the right of a closed curve is its interior.} of $O_L$ and
$O_R$ are subsets of the interior of $O_P$.  Moreover, the
child--parent relation is \emph{well-spaced} if the vertices in
$S(O_P, O_L, O_R)$ have pairwise distance at least $2r$ and, for each
$i \in \{P, L, R\}$, the generalized polygon $O_i$ has distance at
least $r$ to each vertex of $S(O_P, O_L, O_R)$ that is not visited by
$O_i$.  We call the whole polygon hierarchy valid and well-spaced if
each child--parent relation is valid and well-spaced, respectively.

Now consider any independent set $S$ of $G$ with Delaunay complex
$\cD(S)$.  Each noose of a minimum width normalized sphere cut
decomposition of $\cD(S)$ is a generalized polygon and the
nooses of such a sphere cut decomposition form a polygon hierarchy.
Observe that the above definition of being valid is directly derived
from the properties stated in
Lemma~\ref{lem:normalized_noose_combination} and thus this hierarchy
is valid.  Moreover, due to
Lemma~\ref{lem:normalized_nooses_distance}, it is also well-spaced.
Note that the statement of Lemma~\ref{lem:normalized_nooses_distance}
is actually stronger, as it guarantees the distance requirements
globally and not only locally for every child--parent relation.  Thus,
the sphere cut decomposition of $\cD(S)$ yields a valid and
well-spaced polygon hierarchy, leading directly to the following
corollary.
\begin{corollary}
  \label{cor:is_has_noose_hierarchy}
  Let $G$ be a hyperbolic uniform disk graph with radius $r$ and let
  $S$ be an independent set of $G$.  Then the nooses of a normalized
  sphere cut decomposition of $\mathcal{D}(S)$ form a valid and
  well-spaced polygon hierarchy of $G$ whose generalized polygons
  visit all vertices of $S$.
\end{corollary}

Next, we show the converse, i.e., that the vertices visited by
generalized polygons of a valid and well-spaced hierarchy form an
independent set.  Together with the corollary, this shows that finding
a size $k$ independent set of $G$ is equivalent to finding a valid and
well-spaced polygon hierarchy whose generalized polygons visit $k$
vertices.

\begin{figure}[tbp]
  \centering
  \begin{subfigure}[b]{0.30\textwidth}
    \centering
    \includegraphics[page=2]{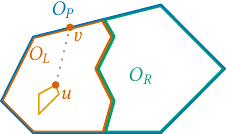}
    \subcaption{}
    \label{fig:nooses_ind:b}
  \end{subfigure}
  \begin{subfigure}[b]{0.30\textwidth}
    \centering
    \includegraphics[page=3]{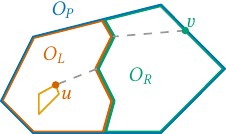}
    \subcaption{}
    \label{fig:nooses_ind:c}
  \end{subfigure}
  \begin{subfigure}[b]{0.30\textwidth}
    \centering
    \includegraphics[page=4]{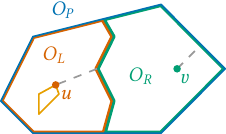}
    \subcaption{}
    \label{fig:nooses_ind:d}
  \end{subfigure}
  \caption{Visualization for the case distinction in the proof of
    Lemma~\ref{lem:noose_hierarchy_yields_is}: either (a) both points
    are in $S(O_{L})$ or $v$ is visited by either (b) $O_{R}$ or (c) a
    descendant of $O_{R}$. The dotted gray line represents a distance
    of at least $2r$, the dashed lines distances of at least $r$.}%
  \label{fig:nooses_ind}
\end{figure}

\begin{lemma}
  \label{lem:noose_hierarchy_yields_is}
  Let $G$ be a hyperbolic uniform disk graph with radius $r$.  The
  vertices visited by generalized polygons of a valid and well-spaced
  polygon hierarchy form an independent set of $G$.
\end{lemma}
\begin{proof}
  Consider a valid and well-spaced polygon hierarchy.  Let $O$ be a
  generalized polygon in this hierarchy and let $S(O)$ be the set of
  vertices of $G$ that are visited by $O$ or by descendants of~$O$.
  We prove the following claim by induction.

  \begin{claim}
    The vertices in $S(O)$ have pairwise distance $2r$ and each vertex
    in $S_O$ not visited by $O$ lies in the interior of $O$ and has
    distance at least $r$ from $O$.
  \end{claim}

  Note that the first part of the claim that the vertices in $S(O)$
  have pairwise distance $2r$ already implies the lemma statement when
  choosing $O$ to be the root of the hierarchy.  The second part of
  the claim is only there to
  enable the induction over the tree structure.
  For the base case, $O$ is a leaf and $S(O)$ is exactly the set of
  vertices visited by $O$.  Moreover, as the hierarchy is well-spaced,
  all vertices visited by $O$ have pairwise distance $2r$.  Hence, the
  claim holds for leaves.

  For the general case, let $O_P$ be a parent polygon with two
  children $O_L$ and $O_R$.  To first get the simpler second part of
  the claim out of the way, let $v \in S(O_P)$ be a vertex not visited
  by $O_P$.  As the child--parent relation is valid, $v$ is either
  visited by $O_L$ and $O_R$ or by a descendant of one of the two.  If
  $v$ is visited by $O_L$ and $O_R$ it lies in the interior of $O_P$
  and has distance at least $r$ to $O_P$ as the child--parent relation
  is well-spaced.  Otherwise, if $v$ is visited by neither $O_L$ nor
  $O_R$ but, without loss of generality, by a descendant of $O_L$.
  Then by induction, it lies in the interior of $O_L$ and has distance
  at least $r$ to $O_L$.  As the interior of $O_L$ is a subset of the
  interior of $O_P$, the same holds for $O_P$.

  It remains to show the first part of the claim, i.e., that any two
  vertices $u, v \in S(O_P)$ have distance at least $2r$.  If $u$ and
  $v$ are both visited by one of the polygons $O_i$ for
  $i \in \{P, L, R\}$, then this follows directly from the fact that
  the child--parent relation is well-spaced.  Otherwise, assume
  without loss of generality that $u$ is not visited by one of these
  three polygons, but is instead visited by a descendant of $O_L$.  We
  distinguish the following three cases of where $v$ lies; see also
  Figure~\ref{fig:nooses_ind}.  The first case is that $v$ is visited by
  $O_L$ or a descendant of $O_L$.  The second case is that it is not
  visited by $O_L$ but by $O_R$ and the third case is that it is
  visited by neither $O_L$ nor $O_R$ but by a descendant of $O_R$.
  Clearly, this covers all possibilities.

  In the first case, $v$ is also visited by $O_L$ or one if its
  descendants.  Thus $u, v \in S(O_L)$ and they have pairwise distance
  at least $2r$ by induction.  In the second case, $v$ lies on $O_R$
  but not on $O_L$ and thus in the exterior of $O_L$.  As the
  hierarchy is well-spaced, $v$ has distance at least $r$ from $O_L$.
  Moreover, by induction, $u$ lies in the interior of $O_L$ and has
  distance at least $r$ from $O_L$.  Thus, $u$ and $v$ have distance
  at least $2r$.  For the third and final case, $v$ lies in the
  interior of $O_R$ and has distance at least $r$ from $O_R$ by
  induction.  As the interiors of $O_L$ and $O_R$ are disjoint, the
  line segment between $u$ and $v$ has to cover the distance of at
  least $r$ from $u$ to $O_L$ and the distance of at least $r$ from
  $O_R$ to $v$.  Thus, also in this final case, $u$ and $v$ have
  distance at least $2r$.
\end{proof}

\subsection{Solving independent set}
\label{sec:exact-algo:solv-indep-set}

Lemma~\ref{lem:noose_hierarchy_yields_is} tells us that for any
hyperbolic uniform disk graph $G$, any valid and well-spaced polygon
hierarchy yields an independent set.  Moreover, by
Corollary~\ref{cor:is_has_noose_hierarchy}, we can obtain any
independent set of $G$ from such a hierarchy.  Thus, finding an
independent set of size $k$ is equivalent to finding a valid and
well-spaced polygon hierarchy that visit $k$ vertices of $G$.  This
can be done using a straightforward dynamic program on the set of all
$k$-candidate nooses or, to be exact, on the not too large superset
due to Lemma~\ref{lem:candidate_nooses_upper_bound}.

The dynamic program processes all $k$-candidate nooses in an order
such that a noose $O$ is processed after a noose $O'$ if the interior
of $O'$ is a subset of the interior of $O$.  When processing a noose
$O$, we compute a valid and well-spaced candidate noose hierarchy with
root $O$ such that the number of vertices visited by nooses in the
hierarchy is maximized.  We call this the \emph{partial solution for
  $O$}.  Note that the maximum over the partial solutions of all
nooses clearly yields a independent set with at least $k$ vertices if
such an independent set exists.  Also observe that when processing
$O$, we have already computed the partial solutions of all possible
child nooses of $O$ as we are only interested in valid hierarchies.
Thus, to compute the partial solution for $O$, it suffices to consider
all pairs of previously processed candidate nooses as potential
children of $O$.  For two such child candidates $O_L$ and $O_R$, we
only need to check whether a child--parent relation with $O$ would be
valid and well-spaced, which can be easily checked by only considering
$O_L$ and $O_R$.  Moreover, if this combination is valid, the number
of vertices visited by the resulting hierarchy with $O$ as root is the
sum of the partial solutions for $O_L$ and $O_R$ minus the vertices
visited by $O_L$ and $O_R$.  Finally, note that the start of the
dynamic program is also easy, as each individual candidate noose by
itself is a valid and well-spaced noose hierarchy (if all its visited
vertices are sufficiently far apart).  This concludes the description
of the dynamic program.  Note that the number of partial solutions and
thus the running time increase excessively if $r$ is chosen
sufficiently small depending on $n$.  In this case, the algorithm is
dominated by an algorithm for Euclidean intersection graphs.  We
obtain the following theorem.

\thmExactalgorithm*
\begin{proof}
  The three upper bounds for the running time follow via different
  algorithms. The first one, which depends on $r$, follows via the
  dynamic program described above. Here, we first enumerate all
  $k$-candidate nooses in $n^{\Oh(1+\twformulak)}$ time
  (Lemma~\ref{lem:candidate_nooses_upper_bound}). Then, we determine
  partial solutions in the form of a polygon hierarchy for each
  $k$-candidate noose, by considering the partial solutions of all
  pairs of smaller polygons. In total, this means that the dynamic
  program can be evaluated in time cubic in the number of
  $k$-candidate nooses to find a valid and well-spaced polygon
  hierarchy maximizing the number of visited vertices in
  $n^{\Oh(1+\twformulak)}$ time. As the vertices visited by a such a
  hierarchy form an independent set by
  Lemma~\ref{lem:noose_hierarchy_yields_is} and the enumerated nooses
  admit a hierarchy corresponding to a size $k$ independent set in $G$
  by Corollary~\ref{cor:is_has_noose_hierarchy} if such an independent
  set exists, this concludes the proof for the first algorithm.

  At the same time \textsc{Independent Set} in $G\in \HUDG(r)$ can
  also be decided in $2^{\Oh(\sqrt{n})}$ time or $n^{\Oh(\sqrt{k})}$
  time.  To see this, recall that the uniformly sized hyperbolic disks
  representing the vertices of $G$ can also be viewed as Euclidean
  disks in the Poincaré disk.  This means that $G$ is also an
  intersection graph of $n$ disks in the Euclidean plane.  As shown by
  De Berg, Bodlaender, Kisfaludi-Bak, Marx, and Van der
  Zanden~\cite[Corollary 2.4]{BergBKMZ20} and by Marx and
  Pilipczuk~\cite{MarxP22}, \textsc{Independent Set} in a disk graph
  can be decided in $2^{\Oh(\sqrt{n})}$ time or $n^{\Oh(\sqrt{k})}$
  time.
\end{proof}

\section{Independent set approximation}
\label{sec:approx_algo_appendix}

Our approximation algorithm in Theorem~\ref{thm:approx-independent-set} works
similarly to Lipton and Tarjan's version for planar graphs~\cite{LiptonT80}: we
repeatedly apply a separator until we get small patches and then solve each
patch individually using the algorithm of Theorem~\ref{thm:exact}. Unlike in
planar graphs, we do not have a strong and readily available lower bound on the
size of a maximum independent set to inform us how much of the graph we can
afford to ignore, so we will start by proving one based on \emph{degeneracy}.

A graph is said to be \emph{$k$-degenerate} if every subgraph has a vertex of
degree at most $k$. By picking this vertex $v$ and recursing on the subgraph
with $v$ and its neighbors removed, we can always get an independent set of
size at least $n/k$. Thus, this gives a lower bound on the size of a maximum
independent set. We will use the following lemma to prove degeneracy.

\begin{lemma}
    \label{lem:degeneracy}
    For any hyperbolic uniform disk graph $G$, there is a vertex whose
    neighborhood can be covered with three cliques and stabbed with four points.
\end{lemma}
\begin{proof}
    Consider the disks of $G$ in the Poincaré disk model. For the remainder of
    the proof we will treat these as Euclidean disks that happen to get smaller
    as they get further from the origin. Take a disk $D$ with maximal
    distance to the origin and without loss of generality, assume its center $c$
    lies on the negative part of the $y$-axis.

    \begin{figure}[tbp]
        \centering
        \hfill
        \begin{subfigure}[t]{0.31\textwidth}
            \centering
            \includegraphics[page=1]{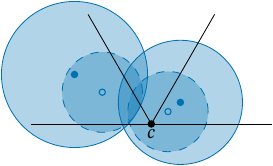}
            \subcaption{Shrinking disks.}
            \label{fig:degeneracy_bound:a}
        \end{subfigure}
        \hfill
        \begin{subfigure}[t]{0.31\textwidth}
            \centering
            \includegraphics[page=2]{figures/degeneracy_bound.pdf}
            \subcaption{Each wedge induces one clique.}
            \label{fig:degeneracy_bound:b}
        \end{subfigure}
        \hfill
        \begin{subfigure}[t]{0.31\textwidth}
            \centering
            \includegraphics[page=3]{figures/degeneracy_bound.pdf}
    \captionsetup{textformat=simple}
            \subcaption{A half-disk can be covered with four disks of half the radius.}
            \label{fig:degeneracy_bound:c}
        \end{subfigure}
        \hfill
        \caption{Figures for the proof of Lemma~\ref{lem:degeneracy}.}
        \label{fig:degeneracy_bound}
    \end{figure}

	Draw a horizontal line through $c$ and two lines at angle $\pi/3$ to form
	six wedges. The lower three wedges cannot contain any disk centers, because
	they would be further from the origin than $c$. Now, shrink each disk
	$D'$ intersecting $D$ until it has the same radius as $D$, while keeping
	the point of $D'$ closest to $c$ fixed; see
	Figure~\ref{fig:degeneracy_bound:a}. The resulting disk will be a subset
	of $D'$ and have its center in the same wedge. This makes the situation as
	in Figure~\ref{fig:degeneracy_bound:b} and means that we can use
	the same arguments as for unit disks \cite{Peeters91}:
	for each wedge, the disks with center in that wedge that intersect $D$ must
	form a clique with $D$, giving three cliques in total.

	Now, we will stab $D$ and the disks it intersects. For this, assume without
	loss of generality that $D$ has unit radius, then consider the problem of
	using unit disks to cover the upper half of a disk of radius $2$ centered
	at $c$. This can be done with four disks, as shown in
	Figure~\ref{fig:degeneracy_bound:c}. The centers of these covering disks
	are our stabbing points:
	any disk $\overline D$ intersecting $D$ must have its center in
	the half-disk and thus $\overline D$ will be stabbed by the center of
	the covering disk that contains the center of $\overline D$.
\end{proof}

By repeatedly picking the vertex given by Lemma~\ref{lem:degeneracy},
the clique bound gives a simple 3-approximation algorithm for
\textsc{Independent Set}. This generalizes the algorithm for unit
disks by Marathe, Breu, Hunt III, Ravi, and
Rosenkrantz~\cite{MaratheBHRR95}. Additionally, we get that any
hyperbolic uniform disk graph is $(3\omega - 3)$-degenerate with
$\omega$ being the size of the largest clique. It is also
$(4\ell - 1)$-degenerate for ply $\ell$, so the maximum independent
set has size $\Omega(n / \ell)$.

As first step in the approximation algorithm, we will use the
separators of De Berg, Bodlaender, Kisfaludi-Bak, Marx, and Van der
Zanden~\cite{BergBKMZ20} to separate the graph into many small
patches\footnote{For certain $r$ it will be better to use the
  separator of Theorem~\ref{thm:separator}, but this only improves
  constants in exponents.}.  These separators are $36/37$-balanced,
have size $\Oh(\sqrt n)$, and can be found in $\Oh(n^4)$ time. The
following proof matches that of Frederickson~\cite{Frederickson87}
that is the main lemma for so-called \emph{$r$-divisions}\footnote{We
  are forced to use $t$ because $r$ already refers to the disk
  radius.}.

\begin{lemma}
	\label{lem:r-division}
	For any $t \geq 1$, we can in $\Oh(n^4 \log(n/t))$ time partition $G$ into
	$\Oh(n / t)$ subgraphs of at most $t$ vertices each and $\Oh(n / \sqrt t)$
	cliques that separate the subgraphs from each other.
\end{lemma}
\begin{proof}
  We will iteratively apply a balanced clique-based separator until
  the size of the graph drops below $t$ vertices.  In particular, we
  use the separator of De Berg, Bodlaender, Kisfaludi-Bak, Marx, and
  Van der Zanden~\cite{BergBKMZ20}, which for a disk graph of $n$
  vertices consists of up to $c \sqrt n$ cliques, for some constant
  $c$. Let $B(n,t)$ denote the number of cliques we need to remove
  from a graph of size $n$ and target size $t$; we will use induction
  to prove that $B(n,t) < 43c n / \sqrt t - 7c \sqrt{n}$ for
  $n > t/37$.  As base case we have $t/37 < n \leq t$. Seeing as
  $n \leq t$, nothing needs to be done, so
  $B(n,t) = 0$. Simultaneously, $n > t/37$ implies
  $43c n / \sqrt t > \frac{43}{\sqrt{47}} c \sqrt{n} > 7c \sqrt{n}$,
  so indeed $B(n,t) = 0 < 43c n / \sqrt t - 7c \sqrt{n}$.
  
  If $n > t$, we use the separator to get two induced subgraphs that are both
  guaranteed to have size at least $n / 37$, then recurse on these.
  This gives the following:
  \begin{align*}
    B(n,t)
    &\leq c \sqrt n + \max_{\alpha \in [\frac{1}{37}, \frac{1}{2}]}
    B(\alpha n, t) + B((1-\alpha)n, t) \\
    &< c \sqrt n + 43c n / \sqrt t - 7c \sqrt{\frac{n}{37}} - 7c \sqrt{\frac{36 n}{37}} \\
    &= 43c n / \sqrt t - 7c \sqrt{n} \cdot \left( \sqrt{1/37} + \sqrt{36/37} - 1/7 \right) \\
    &< 43c n / \sqrt t - 7c \sqrt n.
  \end{align*}
  Thus, $B(n,t) \in \Oh(n / \sqrt t)$.
  The recursion has depth $\Oh(\log(n/t))$ and at each level
  we find separators in disjoint subgraphs of $G$.
  Thus, this procedure takes at most $\Oh(n^4)$ time per level
  and $\Oh(n^4 \log(n/t))$ time in total.
\end{proof}

This gives us all we need to prove Theorem~\ref{thm:approx-independent-set}.
\thmApprox*
\begin{proof}
	Let $t = \ell^2 / \eps^2$ and apply Lemma~\ref{lem:r-division},
	taking $\Oh(n^4 \log n)$ time and removing $\Oh(\eps n / \ell)$ cliques.
	Each of these cliques could only have contributed one point to the
	independent set. Since we know the maximum independent set has size
	$\Omega(n / \ell)$, we can only have thrown away a $\Oh(\eps)$ fraction of
	it, so finding the exact maximum independent set in the remaining graph
	gives a $(1 - \Oh(\eps))$-approximation. We do this by applying
	Theorem~\ref{thm:exact} separately to each of the $\Oh(n/t)$ subgraphs of
	size $t$. This takes $t^{\Oh(1 + \frac{\log t}{r})} =
	(\frac{\ell}{\eps})^{\Oh(1 + \frac{1}{r} \log\frac{\ell}{\eps})}$
	time for each subgraph and thus
	$n \cdot (\frac{\ell}{\eps})^{\Oh(1 + \frac{1}{r} \log\frac{\ell}{\eps})}$
	time in total.
	This makes the total running time $\Oh\left( n^4 \log n\right) + n \cdot
	(\frac{\ell}{\eps})^{\Oh(1 + \frac{1}{r} \log\frac{\ell}{\eps})} $
	to find a $(1 - \Oh(\eps))$-approximation, so by appropriately adjusting
	$\eps$ we can also get a $(1 - \eps)$-approximation with the same
	asymptotic running time.
\end{proof}

Using the results of Har-Peled~\cite{HarPeled23},
Theorem~\ref{thm:approx-independent-set} also directly implies a similarly
efficient approximation algorithm for \textsc{Minimum Vertex Cover}. A set $S$ of vertices in $G$ forms a \emph{vertex cover} if every edge of $G$ is incident to some vertex of $S$. In the \textsc{Minimum Vertex Cover} problem we are given a graph $G$ and $\eps>0$, and we wish to output that a vertex cover $S$ of $G$ that has size $(1+\eps)k^*$, where $k^*$ is the size of the minimum vertex cover.

\begin{corollary}
	Let $\eps\in (0,1)$ and let $G$ be a hyperbolic uniform disk graph
	with radius $r$ and ply $\ell$.
    Then a $(1+\eps)$-approximate minimum vertex cover of $G$ can be computed in
	$\Oh\left( n^4 \log n \right) + n \cdot
	(\frac{\ell}{\eps})^{\Oh(1 + \frac{1}{r} \log\frac{\ell}{\eps})}$
	time.
\end{corollary}

\section{Algorithmic consequences of our separator}
\label{sec:algos_from_sep_appendix}

We can directly use our clique-based separator to compute independent sets and prove Corollary~\ref{cor:indepset_algo} via a simple divide-and conquer algorithm as in Corollary 2.4 of~\cite{BergBKMZ20}. In our case, we use the trivial bound of $n$ on the size of the cliques in Theorem~\ref{thm:separator}, which means that the weight of each clique in the separator is at most $\log(n+1)$. Consequently, the separator's weight is at most $\Oh((1+1/r)\log^2 n)$. The proof of Corollary~\ref{cor:indepset_algo} now follows that of Corollary 2.4 of~\cite{BergBKMZ20}, we merely replace the separator weight $\Oh(n^{1-1/d})$ with $\Oh((1+1/r)\log^2 n)$. We give the proof for completeness.

\corIndepSetAlgo*
\begin{proof}
The algorithm works as follows: First, we find a balanced separator according to Theorem~\ref{thm:separator} in $\Oh(n\log n)$ time. We can then iterate over all possible intersections $X$ between a fixed optimum independent set and the separator. For every clique we can pick one of at most $n$ vertices, or nothing, so there are $(n+1)^{\Oh(1+1/r)\log n} = 2^{\Oh(1+1/r)\log^2 n}$ options. For each of these possible intersections, we remove all disks intersecting $X$, and recurse on both sides of the separator.

The running time follows the recursion $T(n)=2^{\Oh(1+1/r)\log^2 n}T(\frac{2}{3}n)+n^{\Oh(1)}$. Denoting $2^x$ by $\exp_2(x)$, we can unroll this as follows:
\[T(n)=\exp_2\left(\Oh(1+1/r)\sum_{i=0}^{\log_{3/2} n} \log^2 \left(\left(\frac{2}{3}\right)^i n \right)\right)=2^{\Oh(1+1/r)\log^3 n}.\] This concludes the proof.
\end{proof}

We note that the recursion added an extra logarithmic factor, which can be avoided via $\cP$-flattened treewidth. We do not optimize this running time as the resulting algorithm would still be slower than that of Theorem~\ref{thm:exact}. However, we do need to bound the $\cP$-flattened treewidth for Corollary~\ref{cor:algos}.

Consider a graph $G\in \HUDG(r)$ where $r\in \Oh(1)$, let $C$ be any clique of $G$, and let $s$ be the center of some disk of $C$. Then all other disk centers from $C$ are contained in a disk of radius $2r$ around $s$. Since $r\in \Oh(1)$, a disk of radius $r$ has area $\Theta(r^2)$, and the above argument shows that the union of the disks in $C$ occupy area $\Theta(r^2)$. We will now decompose $\Hyp^2$ into regions of diameter at most $2r$ such that each region has an inscribed disk of radius at least $\Omega(r)$. The following can be derived from Lemma 2.1 (ii) of~\cite{Kisfaludi-Bak20} and from Theorem 6 of~\cite{Kisfaludi-BakW24}.

\begin{lemma}[see \cite{Kisfaludi-Bak20,Kisfaludi-BakW24}]
\label{lem:customtile}
For any $\delta\in \Oh(1)$ there exists a subdivision of the hyperbolic plane into connected regions with the following properties:
\begin{enumerate}
\item Each region has diameter less than $\delta$ and thus area $\Oh(\delta^2)$.
\item Each region has an inscribed disk of radius $\Omega(\delta)$, and thus area $\Omega(\delta^2)$,
\end{enumerate}
Moreover, the regions containing some given set of $n$ points can be computed in $\Oh(n)$ time.
\end{lemma}

We are now ready to prove Corollary~\ref{cor:algos}. The proof merely needs to establish a clique partition $\cP$ with a corresponding weighted treewidth bound, and the property that the graph $G_\cP$ has maximum constant degree, that is, each clique of $\cP$ has at most constantly many neighboring cliques. Then the machinery of~\cite{BergBKMZ20} yields the desired algorithms automatically.

\corAlgos*
\begin{proof}
Use the subdivision of Lemma~\ref{lem:customtile} with $\delta=2r$. We can define a clique-partition $\cP$ of $G$ by putting disks in the same partition class whenever their centers fall in the same region of the subdivision. Clearly each created class induces a clique in $G$ as the disk centers have pairwise distance less or equal to the diameter of the region, which is at most $2r$.

Fix such a clique partition $\cP$, and consider now a separator from Theorem~\ref{thm:separator}. If $X$ is one of the cliques in the separator, then let $p_x$ be the center of some disk of $X$. If $C\in \cP$ and $X$ contain disks that intersect each other, then the graph distance of any pair of disks in $C$ and $X$ is at most $3$. Consequently, all disk centers in $C\cup X$ are covered by a disk $D$ of radius $3\cdot 2r=6r$ around $p_x$. Since each region of $\cP$ has area $\Omega(r^2)$, there can be at most $\Oh(1)$ regions of the subdivision contained inside $D$. In particular, $X$ can intersect at most $\Oh(1)$ cliques from $\cP$.

Consider all cliques from $\cP$ that are intersected by the separator cliques of Theorem~\ref{thm:separator}: they clearly form a separator, and the above argument shows that there are $\Oh((1+1/r)\log n)$ such cliques in $\cP$. Since each clique of $\cP$ has weight at most $1+\log n$, we get a separator of weight $\Oh((1+1/r)\log^2 n)$ over $\cP$. Now~\cite{BergBKMZ20} (see also \cite{Kisfaludi-Bak20}) implies that the $\cP$-flattened treewidth of $G$ is $\Oh((1+1/r)\log^2 n)$. Since $\cP$ is a clique-partition with the additional property that each partition class is neighboring to $\Oh(1)$ other cliques, we can apply the entire machinery of~\cite{BergBKMZ20} and derive the required algorithms.
\end{proof}

\section{Conclusion}\label{sec:conclusion}

\subparagraph*{Summary.}
In this article we have explored how the structure of hyperbolic uniform disk graphs depends on their radius~$r$, and showed some algorithmic applications for the \textsc{Independent Set} problem. We proved a clique-based separator theorem, which showed that algorithms for the independent set of unit disks become more efficient as one increases the radius $r$ from the almost-Euclidean setting of $r=1/\sqrt{n}$ to $r=1$, while radii $r>1$ did not give any further gains on the separator size, but maintained a logarithmic separator. This result had some more or less immediate algorithmic consequences that we explored at the end of the paper.

After providing the separator for hyperbolic uniform disk graphs, we studied the Delaunay complexes of point sets that have pairwise distance at least $2r$, and uncovered further separator improvements as $r$ increases from $1$ to $\log n$. We used the outerplanarity bound of such Delaunay complexes to design an exact algorithm for \textsc{Independent Set}. The algorithm is based on dynamic programming using an unknown sphere cut decomposition of the solution's Delaunay complex. The resulting running times became polynomial for $r\in\Omega(\log n)$

Finally, we used a separator-based subdivision together with our exact algorithm to give an approximation scheme for \textsc{Independent Set} for hyperbolic uniform disk graphs of a given ply, further extending our exact algorithm, with only quasi-polynomial dependence on $1/\eps$ and the ply.

\subparagraph*{Future directions.} There are several intriguing future directions to explore in this space. First, is there a fully polynomial approximation scheme for \textsc{Independent Set} for $r\geq 1$? If not, is there at least an approximation scheme that is polynomial in $n$, quasi-polynomial in $1/\eps$, but independent of the ply?

Second, it would be interesting to explore separators in higher-dimensional hyperbolic spaces. The only case studied there is $r=\Theta(1)$~\cite{Kisfaludi-Bak20}. One may expect that separator size should also decrease with growing $r$. Is there a constant $c_d$ such that \textsc{Independent Set} can be solved in polynomial time in $\Hyp^d$ when $r\geq c_d\log n$?

Third, we should explore uniform disk graphs in surfaces of constant curvature, i.e., on the sphere, flat torus, and especially hyperbolic surfaces. What is the size of the best clique-based separator for a uniform disk graph of radius $r$ on a hyperbolic surface of genus~$g$?

Fourth, it would be interesting to investigate the complexity of problems other than \textsc{Independent Set}, and study how their complexity scales with the radius of the disks, or equivalently, with the curvature of the underlying space.

Finally, it is unclear whether our bound on balanced separators is tight for all values of $r$.  For constant $r$, regular hyperbolic tilings are hyperbolic uniform disk graphs with constant clique number and logarithmic treewidth, making our bound asymptotically tight.  However, for larger radii, it could be possible to reduce the logarithmic factor.



\bibliography{references.bib}

\begin{thebibliography}{10}

\bibitem{AgarwalKS98}
Pankaj~K. Agarwal, Marc~J. van Kreveld, and Subhash Suri.
\newblock Label placement by maximum independent set in rectangles.
\newblock {\em Comput. Geom.}, 11(3-4):209--218, 1998.
\newblock \href {https://doi.org/10.1016/S0925-7721(98)00028-5}
  {\path{doi:10.1016/S0925-7721(98)00028-5}}.

\bibitem{Baker94}
Brenda~S. Baker.
\newblock Approximation algorithms for np-complete problems on planar graphs.
\newblock {\em J. {ACM}}, 41(1):153--180, 1994.
\newblock \href {https://doi.org/10.1145/174644.174650}
  {\path{doi:10.1145/174644.174650}}.

\bibitem{BiekerBDJ03}
Nicholas Bieker, Thomas Bl{\"{a}}sius, Emil Dohse, and Paul Jungeblut.
\newblock Recognizing unit disk graphs in hyperbolic geometry is
  $\exists\mathbb{R}$-complete.
\newblock {\em CoRR}, abs/2301.05550, 2023.
\newblock \href {https://arxiv.org/abs/2301.05550} {\path{arXiv:2301.05550}},
  \href {https://doi.org/10.48550/ARXIV.2301.05550}
  {\path{doi:10.48550/ARXIV.2301.05550}}.

\bibitem{BlasiusFFK23}
Thomas Bl{\"{a}}sius, Philipp Fischbeck, Tobias Friedrich, and Maximilian
  Katzmann.
\newblock Solving vertex cover in polynomial time on hyperbolic random graphs.
\newblock {\em Theory Comput. Syst.}, 67(1):28--51, 2023.
\newblock \href {https://doi.org/10.1007/S00224-021-10062-9}
  {\path{doi:10.1007/S00224-021-10062-9}}.

\bibitem{BlasiusFK23}
Thomas Bl{\"{a}}sius, Tobias Friedrich, and Maximilian Katzmann.
\newblock Efficiently approximating vertex cover on scale-free networks with
  underlying hyperbolic geometry.
\newblock {\em Algorithmica}, 85(12):3487--3520, 2023.
\newblock \href {https://doi.org/10.1007/S00453-023-01143-X}
  {\path{doi:10.1007/S00453-023-01143-X}}.

\bibitem{Blasius0KS23}
Thomas Bl{\"{a}}sius, Tobias Friedrich, Maximilian Katzmann, and Daniel
  Stephan.
\newblock Strongly hyperbolic unit disk graphs.
\newblock In {\em 40th International Symposium on Theoretical Aspects of
  Computer Science, {STACS}}, volume 254 of {\em LIPIcs}, pages 13:1--13:17,
  2023.
\newblock \href {https://doi.org/10.4230/LIPICS.STACS.2023.13}
  {\path{doi:10.4230/LIPICS.STACS.2023.13}}.

\bibitem{BlasiusFK16}
Thomas Bl{\"{a}}sius, Tobias Friedrich, and Anton Krohmer.
\newblock Hyperbolic random graphs: Separators and treewidth.
\newblock In {\em 24th Annual European Symposium on Algorithms, {ESA}},
  volume~57 of {\em LIPIcs}, pages 15:1--15:16, 2016.
\newblock \href {https://doi.org/10.4230/LIPICS.ESA.2016.15}
  {\path{doi:10.4230/LIPICS.ESA.2016.15}}.

\bibitem{Effic_Short_Paths_Scale_jour2022}
Thomas Bläsius, Cedric Freiberger, Tobias Friedrich, Maximilian Katzmann,
  Felix Montenegro-Retana, and Marianne Thieffry.
\newblock Efficient shortest paths in scale-free networks with underlying
  hyperbolic geometry.
\newblock {\em ACM Transactions on Algorithms (TALG)}, 18(2):19:1--19:32, 2022.
\newblock \href {https://doi.org/10.1145/3516483} {\path{doi:10.1145/3516483}}.

\bibitem{Struc_Indep_Hyper_Unifo_pre2024}
Thomas Bläsius, Jean-Pierre von~der Heydt, Sándor Kisfaludi-Bak, Marcus
  Wilhelm, and Geert van Wordragen.
\newblock Structure and independence in hyperbolic uniform disk graphs.
\newblock {\em Computing Research Repository (CoRR)}, abs/2407.09362, 2024.
\newblock \href {https://doi.org/10.48550/arXiv.2407.09362}
  {\path{doi:10.48550/arXiv.2407.09362}}.

\bibitem{Bodlaender98}
Hans~L. Bodlaender.
\newblock A partial \emph{k}-arboretum of graphs with bounded treewidth.
\newblock {\em Theor. Comput. Sci.}, 209(1-2):1--45, 1998.
\newblock \href {https://doi.org/10.1016/S0304-3975(97)00228-4}
  {\path{doi:10.1016/S0304-3975(97)00228-4}}.

\bibitem{BouchitteMT03}
Vincent Bouchitt{\'{e}}, Fr{\'{e}}d{\'{e}}ric Mazoit, and Ioan Todinca.
\newblock Chordal embeddings of planar graphs.
\newblock {\em Discret. Math.}, 273(1-3):85--102, 2003.
\newblock \href {https://doi.org/10.1016/S0012-365X(03)00230-9}
  {\path{doi:10.1016/S0012-365X(03)00230-9}}.

\bibitem{BringmannKPL19}
Karl Bringmann, S{\'{a}}ndor Kisfaludi{-}Bak, Michal Pilipczuk, and Erik~Jan
  van Leeuwen.
\newblock On geometric set cover for orthants.
\newblock In {\em 27th Annual European Symposium on Algorithms, {ESA}}, volume
  144 of {\em LIPIcs}, pages 26:1--26:18, 2019.
\newblock \href {https://doi.org/10.4230/LIPICS.ESA.2019.26}
  {\path{doi:10.4230/LIPICS.ESA.2019.26}}.

\bibitem{Chan03}
Timothy~M. Chan.
\newblock Polynomial-time approximation schemes for packing and piercing fat
  objects.
\newblock {\em J. Algorithms}, 46(2):178--189, 2003.
\newblock \href {https://doi.org/10.1016/S0196-6774(02)00294-8}
  {\path{doi:10.1016/S0196-6774(02)00294-8}}.

\bibitem{fptbook}
Marek Cygan, Fedor~V Fomin, {\L}ukasz Kowalik, Daniel Lokshtanov, D{\'a}niel
  Marx, Marcin Pilipczuk, Micha{\l} Pilipczuk, and Saket Saurabh.
\newblock {\em Parameterized Algorithms}.
\newblock Springer, 2015.

\bibitem{BergBKMZ20}
Mark de~Berg, Hans~L. Bodlaender, S{\'{a}}ndor Kisfaludi{-}Bak, D{\'{a}}niel
  Marx, and Tom~C. van~der Zanden.
\newblock A framework for exponential-time-hypothesis-tight algorithms and
  lower bounds in geometric intersection graphs.
\newblock {\em {SIAM} J. Comput.}, 49(6):1291--1331, 2020.
\newblock \href {https://doi.org/10.1137/20M1320870}
  {\path{doi:10.1137/20M1320870}}.

\bibitem{BergKMT23}
Mark de~Berg, S{\'{a}}ndor Kisfaludi{-}Bak, Morteza Monemizadeh, and Leonidas
  Theocharous.
\newblock Clique-based separators for geometric intersection graphs.
\newblock {\em Algorithmica}, 85(6):1652--1678, 2023.
\newblock \href {https://doi.org/10.1007/S00453-022-01041-8}
  {\path{doi:10.1007/S00453-022-01041-8}}.

\bibitem{Dorn_Penninkx_Bodlaender_Fomin_2010}
Frederic Dorn, Eelko Penninkx, Hans~L. Bodlaender, and Fedor~V. Fomin.
\newblock Efficient exact algorithms on planar graphs: Exploiting sphere cut
  decompositions.
\newblock {\em Algorithmica}, 58(3):790--810, 2010.
\newblock \href {https://doi.org/10.1007/s00453-009-9296-1}
  {\path{doi:10.1007/s00453-009-9296-1}}.

\bibitem{Frederickson87}
Greg~N. Frederickson.
\newblock Fast algorithms for shortest paths in planar graphs, with
  applications.
\newblock {\em {SIAM} J. Comput.}, 16(6):1004--1022, 1987.
\newblock \href {https://doi.org/10.1137/0216064} {\path{doi:10.1137/0216064}}.

\bibitem{Diamet_Hyper_Random_Graph-FriedKrohm18}
Tobias Friedrich and Anton Krohmer.
\newblock On the diameter of hyperbolic random graphs.
\newblock {\em SIAM Journal on Discrete Mathematics}, 32(2):1314--1334, 2018.
\newblock \href {https://doi.org/10.1137/17M1123961}
  {\path{doi:10.1137/17M1123961}}.

\bibitem{Euclid_Non_Euclid_Geomet-Green93}
Marvin~Jay Greenberg.
\newblock {\em Euclidean and Non-Euclidean Geometries: Development and
  History}.
\newblock W. H. Freeman and Company, 3rd edition, 1993.

\bibitem{Random_Hyper_Graph-Gugel12}
Luca Gugelmann, Konstantinos Panagiotou, and Ueli Peter.
\newblock Random hyperbolic graphs: Degree sequence and clustering.
\newblock In {\em 39th International Colloquium on Automata, Languages, and
  Programming (ICALP)}, pages 573--585, 2012.
\newblock \href {https://doi.org/10.1007/978-3-642-31585-5_51}
  {\path{doi:10.1007/978-3-642-31585-5_51}}.

\bibitem{Har-Peled14}
Sariel Har{-}Peled.
\newblock Quasi-polynomial time approximation scheme for sparse subsets of
  polygons.
\newblock In {\em 30th Annual Symposium on Computational Geometry, {SoCG}},
  page 120. {ACM}, 2014.
\newblock \href {https://doi.org/10.1145/2582112.2582157}
  {\path{doi:10.1145/2582112.2582157}}.

\bibitem{HarPeled23}
Sariel Har{-}Peled.
\newblock Approximately: Independence implies vertex cover.
\newblock {\em CoRR}, abs/2308.00840, 2023.
\newblock \href {https://arxiv.org/abs/2308.00840} {\path{arXiv:2308.00840}},
  \href {https://doi.org/10.48550/ARXIV.2308.00840}
  {\path{doi:10.48550/ARXIV.2308.00840}}.

\bibitem{HochbaumM85}
Dorit~S. Hochbaum and Wolfgang Maass.
\newblock Approximation schemes for covering and packing problems in image
  processing and {VLSI}.
\newblock {\em J. {ACM}}, 32(1):130--136, 1985.
\newblock \href {https://doi.org/10.1145/2455.214106}
  {\path{doi:10.1145/2455.214106}}.

\bibitem{ImpagliazzoP01}
Russell Impagliazzo and Ramamohan Paturi.
\newblock On the complexity of $k$-{SAT}.
\newblock {\em Journal of Computer and System Sciences}, 62(2):367--375, 2001.
\newblock \href {https://doi.org/10.1006/jcss.2000.1727}
  {\path{doi:10.1006/jcss.2000.1727}}.

\bibitem{centerpoint}
Shreesh Jadhav and Asish Mukhopadhyay.
\newblock Computing a centerpoint of a finite planar set of points in linear
  time.
\newblock {\em Discret. Comput. Geom.}, 12:291--312, 1994.
\newblock \href {https://doi.org/10.1007/BF02574382}
  {\path{doi:10.1007/BF02574382}}.

\bibitem{Kisfaludi-Bak20}
S{\'{a}}ndor Kisfaludi{-}Bak.
\newblock Hyperbolic intersection graphs and (quasi)-polynomial time.
\newblock In {\em Proceedings of the 2020 {ACM-SIAM} Symposium on Discrete
  Algorithms, {SODA}}, pages 1621--1638. {SIAM}, 2020.
\newblock \href {https://doi.org/10.1137/1.9781611975994.100}
  {\path{doi:10.1137/1.9781611975994.100}}.

\bibitem{Kisfaludi-Bak21}
S{\'{a}}ndor Kisfaludi{-}Bak.
\newblock A quasi-polynomial algorithm for well-spaced hyperbolic {TSP}.
\newblock {\em J. Comput. Geom.}, 12(2):25--54, 2021.
\newblock \href {https://doi.org/10.20382/JOCG.V12I2A3}
  {\path{doi:10.20382/JOCG.V12I2A3}}.

\bibitem{Kisfaludi-BakW24}
S{\'{a}}ndor Kisfaludi{-}Bak and Geert van Wordragen.
\newblock A quadtree, a steiner spanner, and approximate nearest neighbours in
  hyperbolic space.
\newblock In {\em 40th International Symposium on Computational Geometry,
  {SoCG}}, volume 293 of {\em LIPIcs}, pages 68:1--68:15, 2024.
\newblock \href {https://doi.org/10.4230/LIPICS.SOCG.2024.68}
  {\path{doi:10.4230/LIPICS.SOCG.2024.68}}.

\bibitem{koebe1936kontaktprobleme}
Paul Koebe.
\newblock {\em Kontaktprobleme der konformen Abbildung}.
\newblock Hirzel, 1936.

\bibitem{Hyper_Geomet_Compl_Networ-Kriouk10}
Dmitri Krioukov, Fragkiskos Papadopoulos, Maksim Kitsak, Amin Vahdat, and
  Marián Boguñá.
\newblock Hyperbolic geometry of complex networks.
\newblock {\em Phys. Rev. E}, 82:036106, 2010.
\newblock \href {https://doi.org/10.1103/PhysRevE.82.036106}
  {\path{doi:10.1103/PhysRevE.82.036106}}.

\bibitem{LiptonT79}
Richard~J. Lipton and Robert~Endre Tarjan.
\newblock A separator theorem for planar graphs.
\newblock {\em SIAM Journal on Applied Mathematics}, 36(2):177--189, 1979.

\bibitem{LiptonT80}
Richard~J. Lipton and Robert~Endre Tarjan.
\newblock Applications of a planar separator theorem.
\newblock {\em SIAM Journal on Computing}, 9(3):615--627, 1980.

\bibitem{MaratheBHRR95}
Madhav~V. Marathe, H.~Breu, Harry B.~Hunt III, S.~S. Ravi, and Daniel~J.
  Rosenkrantz.
\newblock Simple heuristics for unit disk graphs.
\newblock {\em Networks}, 25(2):59--68, 1995.
\newblock \href {https://doi.org/10.1002/NET.3230250205}
  {\path{doi:10.1002/NET.3230250205}}.

\bibitem{Marx07a}
D{\'{a}}niel Marx.
\newblock On the optimality of planar and geometric approximation schemes.
\newblock In {\em 48th Annual {IEEE} Symposium on Foundations of Computer
  Science ({FOCS} 2007)}, pages 338--348. {IEEE} Computer Society, 2007.
\newblock \href {https://doi.org/10.1109/FOCS.2007.26}
  {\path{doi:10.1109/FOCS.2007.26}}.

\bibitem{MarxPP18}
D{\'{a}}niel Marx, Marcin Pilipczuk, and Michal Pilipczuk.
\newblock On subexponential parameterized algorithms for steiner tree and
  directed subset {TSP} on planar graphs.
\newblock In {\em 59th {IEEE} Annual Symposium on Foundations of Computer
  Science, {FOCS}}, pages 474--484. {IEEE} Computer Society, 2018.
\newblock \href {https://doi.org/10.1109/FOCS.2018.00052}
  {\path{doi:10.1109/FOCS.2018.00052}}.

\bibitem{MarxP22}
D{\'{a}}niel Marx and Michal Pilipczuk.
\newblock Optimal parameterized algorithms for planar facility location
  problems using {V}oronoi diagrams.
\newblock {\em {ACM} Trans. Algorithms}, 18(2):13:1--13:64, 2022.
\newblock \href {https://doi.org/10.1145/3483425} {\path{doi:10.1145/3483425}}.

\bibitem{McDiarmidM13}
Colin McDiarmid and Tobias M{\"{u}}ller.
\newblock Integer realizations of disk and segment graphs.
\newblock {\em J. Comb. Theory {B}}, 103(1):114--143, 2013.
\newblock \href {https://doi.org/10.1016/J.JCTB.2012.09.004}
  {\path{doi:10.1016/J.JCTB.2012.09.004}}.

\bibitem{MillerTTV97}
Gary~L. Miller, Shang{-}Hua Teng, William~P. Thurston, and Stephen~A. Vavasis.
\newblock Separators for sphere-packings and nearest neighbor graphs.
\newblock {\em Journal of the {ACM}}, 44(1):1--29, 1997.
\newblock \href {https://doi.org/10.1145/256292.256294}
  {\path{doi:10.1145/256292.256294}}.

\bibitem{Diamet_KPKVB_Random_Graph-MüllerStaps19}
Tobias Müller and Merlijn Staps.
\newblock The diameter of {KPKVB} random graphs.
\newblock {\em Advances in Applied Probability}, 51(2):358--377, 2019.
\newblock \href {https://doi.org/10.1017/apr.2019.23}
  {\path{doi:10.1017/apr.2019.23}}.

\bibitem{Peeters91}
M.J.P. Peeters.
\newblock {\em On coloring j-unit sphere graphs}, volume FEW 512 of {\em
  Research memorandum / Tilburg University, Department of Economics}.
\newblock 1991.
\newblock Pagination: 6, vi.

\bibitem{ratcliffe_foundations_2019}
John~G. Ratcliffe.
\newblock {\em Foundations of {Hyperbolic} {Manifolds}}, volume 149 of {\em
  Graduate {Texts} in {Mathematics}}.
\newblock Springer International Publishing, 2019.
\newblock \href {https://doi.org/10.1007/978-3-030-31597-9}
  {\path{doi:10.1007/978-3-030-31597-9}}.

\bibitem{Robertson_Seymour_1991}
Neil Robertson and P.~D Seymour.
\newblock Graph minors. x. obstructions to tree-decomposition.
\newblock {\em Journal of Combinatorial Theory, Series B}, 52(2):153--190,
  1991.
\newblock \href {https://doi.org/10.1016/0095-8956(91)90061-N}
  {\path{doi:10.1016/0095-8956(91)90061-N}}.

\bibitem{DBLP:journals/combinatorica/SeymourT94}
Paul~D. Seymour and Robin Thomas.
\newblock Call routing and the ratcatcher.
\newblock {\em Comb.}, 14(2):217--241, 1994.
\newblock \href {https://doi.org/10.1007/BF01215352}
  {\path{doi:10.1007/BF01215352}}.

\bibitem{EijkhofBK07}
Frank van~den Eijkhof, Hans~L. Bodlaender, and Arie M. C.~A. Koster.
\newblock Safe reduction rules for weighted treewidth.
\newblock {\em Algorithmica}, 47(2):139--158, 2007.
\newblock \href {https://doi.org/10.1007/S00453-006-1226-X}
  {\path{doi:10.1007/S00453-006-1226-X}}.

\end{thebibliography}

\appendix

\end{document}